\documentclass[a4paper]{elsarticle}
\usepackage{amsfonts,amssymb,amsthm,amsmath,url}
\usepackage{subfigure,xspace}

\def\compile{compile}

\usepackage{graphicx, pgf, tikz}
\usetikzlibrary{arrows}
\usetikzlibrary{decorations.pathmorphing}
\usetikzlibrary{patterns,fadings}

\pgfdeclarepatternformonly{big north east lines}{\pgfqpoint{-2pt}{-2pt}}{\pgfqpoint{8pt}{8pt}}{\pgfqpoint{6pt}{6pt}}
{
  \pgfsetlinewidth{0.4pt}
  \pgfpathmoveto{\pgfqpoint{0pt}{0pt}}
  \pgfpathlineto{\pgfqpoint{6.2pt}{6.2pt}}
  \pgfusepath{stroke}
}

\newcommand\perio[1]{p_{#1}}
\newcommand\lediff[1]{\delta_l(#1)}
\newcommand\ridiff[1]{\delta_r(#1)}

\newcommand\putcolor{}

\newcommand\cc[1]{{\putcolor{}\ifmmode\textsc{CC}\left({#1}\right)\else$\textsc{cc}\left({#1}\right)$\fi}}
\newcommand\predp[1]{{\putcolor{}\ifmmode\textsc{Pred}_{#1}
    \else$\textsc{Pred}_{#1}$\fi}}
\newcommand\invap[2]{{\putcolor{}\ifmmode\textsc{Inv}_{#1}^{#2}
    \else$\textsc{Inv}_{#1}^{#2}$\fi}}
\newcommand\cyclp[2]{{\putcolor{}\ifmmode\textsc{Cycle}_{#1}^{#2}\else
  $\textsc{Cycle}_{#1}^{#2}$\fi}}

\newcommand\predc[1]{{\putcolor{}\cc{\predp{#1}}}}
\newcommand\invac[2]{{\putcolor{}\cc{\invap{#1}{#2}}}}
\newcommand\cyclc[2]{{\putcolor{}\cc{\cyclp{#1}{#2}}}}

\newcommand{\setN}{\mathbb{N}}
\newcommand{\setZ}{\mathbb{Z}}

\newcommand\TODO[1]{{\color{red}#1}}

\newcommand{\alphA}{Q_F}
\newcommand{\alphB}{Q_G}
\newcommand\ZZ{\setZ}

\newcommand\bloc[1]{b_{#1}}
\newcommand\debloc[1]{b^{-1}_{#1}}
\newcommand\sac{\sqsubseteq}
\newcommand{\simu}{\preccurlyeq}
\newcommand\grp[2]{{#1}^{<#2>}}

\newcommand\moindre{\prec}

\newcommand\CC{\mathbf{cc}}

\newcommand\IP{{\putcolor{}\textsc{ip}}\xspace}
\newcommand\DISJ{{\putcolor{}\textsc{disj}}\xspace}
\newcommand\EQ{{\putcolor{}\textsc{eq}}\xspace}

\newtheorem{definition}{Definition}
\newtheorem{proposition}{Proposition}

\newtheorem*{remark}{Remark}
\newtheorem*{claim}{Claim}
\newtheorem{lemma}{Lemma}
\newtheorem{corollary}{Corollary}

\newcommand{\singlerule}[4]{%
      \begin{tabular}[c]{c@{\hskip .2em}c@{\hskip .2em}c}
        &#1\\
        #2&#3&#4
      \end{tabular}
}

\newcommand{\carule}[8]{%
    \small
    \begin{tabular}[c]{cccccccc}
        \singlerule#1 000
&       \singlerule#2 001
&       \singlerule#3 010
&       \singlerule#4 011
&       \singlerule#5 100
&       \singlerule#6 101
&       \singlerule#7 110
&       \singlerule#8 111
    \end{tabular}%
}
\DeclareMathOperator{\overoplus}{\overline\oplus}

\ifx\defined\compile
\newcommand\docompile[1]{\TODO{Ce dessin est trop long a compiler}}
\else
\newcommand\docompile[1]{#1}
\fi

\newcommand\ignore[1]{}

\begin{document}
\begin{frontmatter}
  \title{Communication Complexity\\ and Intrinsic Universality\\ in Cellular
 Automata\footnote{Partially supported by programs Fondap and Basal-CMM, 
Fondecyt 1070022 (E.G) and Fondecyt 1090156 (I.R.).}}

  \author[uai]{E. Goles}
  \author[lama]{P.-E. Meunier}
  \author[dim]{I. Rapaport}
  \author[lama]{G. Theyssier\corref{cora}} 
  \cortext[cora]{Corresponding author
    (\url{guillaume.theyssier@univ-savoie.fr})}

  \address[uai]{Facultad de Ingenieria y Ciencias, Universidad Adolfo
    Ib\'a\~nez, Santiago, Chile}
  \address[dim]{DIM, CMM (UMI 2807 CNRS), Universidad de Chile, Santiago, Chile}
  \address[lama]{LAMA, Universit\'e de Savoie, CNRS,
    73\hspace{0.2em}376 Le Bourget-du-Lac Cedex, France}

\begin{abstract}

  The notions of universality and completeness are central in the theories of
  computation and computational complexity.  However, proving lower bounds and
  necessary conditions remains hard in most of the cases. In this article, we
  introduce necessary conditions for a cellular automaton to be ``universal'',
  according to a precise notion of simulation, related both to the dynamics of
  cellular automata and to their computational power. This notion of simulation
  relies on simple operations of space-time rescaling and it is intrinsic to the
  model of cellular automata. \emph{Intrinsic universality}, the derived
  notion, is stronger than Turing universality, but more uniform, and easier to
  define and study.

  Our approach builds upon the notion of \emph{communication complexity},
  which was primarily designed to study parallel programs, and thus
  is, as we show in this article, particulary well suited to
  the study of cellular automata: it allowed us to show, by studying
  natural problems on the dynamics of cellular automata, that
  several classes of cellular automata, as well as many natural
  (\emph{elementary}) examples, were not
  \emph{intrinsically universal}.

\ignore{
  Let $F$ be a cellular automaton (CA).  This paper establishes
  necessary conditions for $F$ in order to be intrinsically
  universal. The central idea is to consider the communication
  complexity of various ``canonical problems'' related to the dynamics
  of $F$. We show that the intrinsic universality of $F$ implies high
  communication complexity for each of the canonical problems.  This
  result allows us to rule out many CA from being intrinsically
  universal: The linear CA, the expansive CA, the reversible CA and
  the elementary CA 218, 33 and 94.  The notion of intrinsic universality
  is related to a process by which we change the scale of  space-time diagrams.
  Therefore,  in this work we are answering pure dynamical question by 
  using a computational theory. This communication complexity theory,
  on the other hand,   provides a finer tool than the one given by classical computational
  complexity analysis.  In fact, we prove that for two of the
  canonical problems there exists a CA for which the computational
  complexity is maximal (\textsc{P}-complete, or $\Pi_1^0$-complete)
  while the corresponding communication complexity is rather low.  We
  also show the orthogonality of the problems. More precisely, for any
  pair of problems there exists a CA with low communication complexity
  for one but high communication complexity for the other.
}
\end{abstract}

\end{frontmatter}

\section{Introduction}
\label{sec:intro}

Since the pioneering work of J.~von~Neumman \cite{neumann67}, universality in
cellular automata (CA) has received a lot of attention
(see~\cite{surveyOllinger} for a survey). Historically, the notion of
universality used for CA was more or less an adaptation of the classical
Turing-universality. Later, a stronger notion called \emph{intrinsic
  universality} was proposed: A CA is intrinsically universal if it is able to
simulate any other CA~\cite{liferokadur,rapaport99,surveyOllinger} through a
uniform and regular encoding based on \emph{rescaling}.

This definition of intrinsic universality may seem very restrictive.  However,
it can be very common among natural families of CA~\cite{BoyerT09}, and allows a
complete and precise formalization of the notion of universality \footnote{There
  is actually no consensus on the formal definition of Turing-universality in CA
  (see~\cite{liferokadur} for a discussion about encoding/decoding
  problems).}. As we are going to see, this preciseness, and the robustness of
this definition, allows for concrete proofs of negative results and lower
bounds.

Indeed, in this paper we will explain how to rule out
particular elementary cellular automata,
as well as whole well-known classes of cellular automata, 
from being intrinsically universal,
using the elegant framework of communication complexity.

In Section~\ref{sec:def} we give the basic definitions.
One of the key definitions is the following:
Given a traditional computational problem ${\cal P}$ with an
arbitrary input $w$, we can split the input into two subwords
$w_1$ and $w_2$; therefore, we can refer to the ``comunication complexity'' 
of such problem ($w_1$ is given to Alice while $w_2$ is given to Bob). 

In Section~\ref{sec:probs} we introduce a
family of ``canonical problems'' concerning various aspects of the
dynamics of a given CA. In other words, for any CA $F$ and any 
prototype problem ${\cal P}$, we consider the problem ${\cal P}_{F}$.

In  Section~\ref{sec:univ} we
explain how to infer properties of $F$ 
from the study of the communication complexity
of ${\cal P}_F$. More precisely, we prove that  
if the communication complexity 
of one of our canonical problem ${\cal P}_{F}$
is not maximal, 
then $F$ is not intrinsically universal.
In other words, we are introducing a powerful tool 
for ruling out CA
from being intrinsically universal. We conclude
that linear, expansive and reversible CA are
not intrinsically universal. We also 
show the uncomparability of our three canonical problems: none
of them is sufficient to discard \emph{all} non-universal 
cellular automata, and none of them is stronger than any other.

In Section~\ref{sec:iu} we explain clearly why the communication
complexity approach appears to be a promising tool for ruling out CA
from being intrinsically universal. More precisely, we prove
computational intractability results about problems that our
framework considers very simple.

Finally, in Section~\ref{sec:concrete} we use our results to prove
that a few concrete elementary CA are not intrinsically
universal. Although looking at several space-time diagrams of these
automata might give a strong intuition about their non-universality, we
stress that producing complete formal proofs for such a negative
result is a difficult task and, as far as we know, had never been done
before.

\section{Basic definitions}
\label{sec:def}

\subsection{Communication complexity}

Communication complexity is a notion introduced by A.~C.-C.~Yao 
in \cite{yao79},
and designed at first for lower-bounding the amount of communication 
needed in
distributed algorithms. In that model he considered two players, 
namely Alice and Bob,
both with arbitrary computational power and communicating to 
each other in order to collaboratively decide the
value of a given function. More precisely,  for a function 
$\phi: X\times Y\rightarrow Z$, the
question  is ``how much information do they need to exchange, 
in the worst case, in order to compute $\phi(x,y)$, with
Alice knowing only $x$ and Bob only $y$''.

This communication problem is solved by a \emph{protocol}, which 
specifies, at each step of the
communication between Alice and Bob, who speaks (Alice or Bob), and what she/he
says (a bit, 0 or 1), as a function of her/his respective input. 
This simple framework, and some of its variants we discuss in this article, 
appears  to be promising for studying CA. 

A protocol $\mathcal{P}$ over a domain $X\times Y$ with range $Z$ is a binary
tree where each internal node $v$ is labeled either by a map
$a_v:X\rightarrow \{0,1\}$ or by a map $b_v:Y\rightarrow \{0,1\}$, and each leaf
$v$ is labeled either by a map ${A_v:X\rightarrow Z}$ or by a map
${B_v:Y\rightarrow Z}$.

The \emph{value} of protocol $\mathcal{P}$ on input $(x,y)\in X\times Y$ is
given by $A_v(x)$ (or $B_v(y)$) where $A_v$ (or $B_v$) is the label of the leaf
reached by walking on the tree from the root, turning left if $a_v(x)=0$ (or
$b_v(y)=0$), and right otherwise. We say that a protocol computes a
function ${\phi: X\times Y\rightarrow Z}$ if for any $(x,y)\in X\times Y$, its value
on input $(x,y)$ is $\phi(x,y)$.

Intuitively, each internal node specifies a bit to be communicated
either by Alice or by Bob, whereas at the leaves either Alice or Bob
determines the value of $\phi$ when she/he has received enough
information from the other party.

In our formalism, we do not ask both Alice and Bob to be able to
give the final value. We do so in order to consider protocols
where communication is unidirectional.

We denote by $\CC(\phi)$ the (deterministic) communication complexity
of a function $\phi:X\times Y\rightarrow Z$. It is the minimal cost of
a protocol, over
all protocols computing $\phi$, where the cost of a protocol is the 
depth of its corresponding tree.


One approach for
proving lower bounds on the communication complexity of 
an arbitrary function $\phi$
is based on the so-called
\emph{fooling sets} (for a deeper presentation of 
this theory we refer to \cite{kushilevitz97}).

\begin{definition}
\label{def:fool}
Given a function ${\phi: X\times Y\rightarrow Z}$, a set ${S\subseteq
  X\times Y}$ is a \emph{fooling set} for $\phi$ if there exists $z\in Z$
with:
\begin{enumerate}
\item ${\forall (x,y)\in S,\, \phi(x,y)=z}$,
\item ${\forall (x_1,y_1)\in S,\forall (x_2,y_2)\in S}$, either
  ${\phi(x_1,y_2)\not=z}$ or ${\phi(x_2,y_1)\not=z}$.
\end{enumerate}
\end{definition}

The usefulness of fooling sets is given by the following lemma (see
\cite{kushilevitz97}).

\begin{lemma}
  If $S$ is a fooling set of size $m$ for $\phi$ then
  $\CC(\phi)\geq\log_2(m)$.
\end{lemma}

In addition to ad hoc fooling set constructions, we will use the
following classical lower bounds on communication complexity (the proofs
appear in \cite{kushilevitz97}).

\begin{proposition}
  \label{prop:commlb}
  Let $n\geq 1$ be fixed.  Let $\phi_{\EQ}$, $\phi_{\IP}$ and $\phi_{\DISJ}$ be
  the functions ``equality'', ``inner product'' and ``disjointness'' defined
  from ${\{0,1\}^n\times\{0,1\}^n}$ to ${\{0,1\}}$ by:
  \begin{align*}
    \phi_{\EQ} (x,y) &=
    \begin{cases}
      1&\text{ if }(\forall i) (x_i=y_i),\\
      0&\text{ otherwise.}
    \end{cases}\\
    \phi_{\IP} (x,y) &=
    \begin{cases}
      1&\text{ if }\sum_i x_iy_i\bmod 2 = 1,\\
      0&\text{ otherwise.}
    \end{cases}\\
    \phi_{\DISJ} (x,y) &=
    \begin{cases}
      1&\text{ if }(\forall i) (x_iy_i\not=1),\\
      0&\text{ otherwise.}
    \end{cases}
  \end{align*}
  The following lower bounds hold:
  \begin{itemize}
  \item $\CC(\phi_{\EQ})\geq n$.
  \item $\CC(\phi_{\IP})\geq n$.
  \item $\CC(\phi_{\DISJ})\geq n$.\\
  \end{itemize}
\end{proposition}

\subsection{Splitting the input of computational problems}

Let us consider now classical computational input-output problems.
In this work we will only encounter problems of the form
${\cal P}:Q^\ast\rightarrow Z$, whose inputs are words over
some alphabet $Q$ and outputs are elements of a finite set $Z$. 
Moreover, we will always have $Z=Q$ or $Z=\{0,1\}$ as output sets. 

Given such type of problem $\cal P$,
we define, for any $n$,  its restriction to words of length $n$; i.e, 
we consider the restricted problem ${{\cal P}|_n:Q^{n}\rightarrow Z}$.

The key idea of the communication approach is to {\emph{split}} the
input into two parts: For any ${1 \leq i \leq (n-1)}$, we define 
${{\cal P}|^i_n : Q^i\times Q^{n-i}\rightarrow Z}$. More 
precisely, for every $x \in Q^i, y \in  Q^{n-i}$, we have
${\cal P}|^i_n(x,y)={\cal P}|_n(xy)$. Then, we can consider
the communication complexity $\CC\bigl({\cal P}|^i_n\bigr)$ of the $i$th
split function ${\cal P}|^i_n$. Of course the choice of $i$ matters
and can alter the corresponding communication complexity. Since we
don't want to rely on an arbitrary choice, we consider the worst
case. This yields the following definition:

\begin{definition}
Let ${\cal P}:Q^\ast\rightarrow Z$ be a problem. The
\emph{communication complexity} of $\cal P$, denoted $\cc{\cal P}$,
is the function:
\[ n\mapsto \max_{1 \leq i\leq n-1} \CC\bigl({\cal P}|_n^i\bigr).\]
\end{definition}

\subsection{Cellular automata}

In this paper we are always going to consider one-dimensional CA.  A
CA is defined by its local rule $f:Q^{2r+1} \rightarrow Q$ (where $Q$
corresponds to the set of states and $r$ denotes the radius of the
local rule).  For any ${n\geq 2r+1}$, we extend $f:Q^{2r+1} \rightarrow Q$ 
to the more general
$f:Q^n\rightarrow Q^{n-2r}$ by
\[f(u_1\cdots u_n) = f(u_1\cdots u_{2r+1})\cdots f(u_{n-2r}\cdots
u_n).\] Moreover, for every $1 \leq t \leq \lfloor{(n-1)/2r}\rfloor$,
we define the $t$-steps local iteration as $f^t:Q^{n} \rightarrow
Q^{n-2\cdot r\cdot t}$ by
\[\begin{cases}
  f^1 = f\\
  f^t(u_1\cdots u_n) = f\bigl(f^{t-1}(u_1\cdots u_{n-2r})\cdots f^{t-1}(u_{2r+1}\cdots u_n)\bigr)
\end{cases}\]
We also define ${f^\ast:Q^\ast\rightarrow Q^\ast}$ by
\[f^\ast(u) = f^{\left\lfloor\frac{|u|-1}{2r}\right\rfloor} (u).\]

Intuitively, $f^\ast$ applied on $u$ consist in iterating $f$ as long
as possible (until ending up with a word too short for $f$). The result is
a word of length at most $2r$ (depending on ${|u| \bmod 2r}$).

We denote by $F:Q^\setZ \rightarrow Q^\setZ$ the global rule induced
by $f$ following the classical definition:
\[F(c)_z = f(c_{z-r},\ldots,c_{z+r}).\]

Finally, we denote by $F^t:Q^\setZ \rightarrow Q^\setZ$ the $t$-step
iteration of the global function $F$.

A  global function $F$ can be represented by different local
functions. All properties considered in this paper depend only on $F$
and are not sensitive to the choice of a particular local
function. However, to avoid useless formalism, we will use the
following notion of \emph{canonical} local representation: $(f,r)$ is
the canonical local representation of $F$ if $f$ has radius $r$ and it is
the local function of smallest radius having $F$ as its associated global
function.

Throughout this work we are going to refer to the CA $F$ with $(f,r)$
being its canonical local representation. 

\section{The three canonical communication problems}
\label{sec:probs}

In this section we define the three ``problem schemes'' on which we are going to
apply the communication complexity approach. Before entering into
details, we stress that this set of problems tackles various dynamical
aspects of CA: Transient, periodic and asymptotic regime starting
respectively from finite, cyclic, or ultimately periodic
configurations. Moreover, algorithmically speaking, they are also very
different since they belong respectively to the classes \textsc{p},
\textsc{pspace}, and $\Pi_1^0$ (and can be complete for these
classes as we will see in this section).

Thus, they form an interesting set of prototype problems.

\subsection{Prediction}
\label{sec:pred}

The prediction problem consists in determining the far future of a
cell given the state of sufficiently many cells around it.

\begin{definition}
  \label{def:prediction}

Let $F$ be a CA.  The problem $\predp{F}: Q^{\ast} \to Q$ is defined
as follows: \[\predp{F}(u)= \bigl(f^\ast(u)\bigr)_1,\] where
$(f,r)$ is the canonical local representation of $F$  while the
``$\bigl(f^\ast(u)\bigr)_1$'' notation means that we take the first letter
of the word $f^\ast(u)$, which has length at most $2r$.
\end{definition}

Clearly, this problem is in $\textsc{DTime}(n^2)$, and,
as we have already said before, we can also  view \predp{F} as a
communication problem (see Figure \ref{fig:pred}): Given an initial
configuration as input, we {\emph{split}} the initial configuration
between Alice and Bob, and ask for the \emph{final} value computed by
$F$ on this input configuration, as represented in Figure
\ref{fig:pred.cc}.

\begin{figure}[ht]
\begin{center}
\subfigure[A space-time diagram of Rule 110.]
{
\begin{tikzpicture}[scale=0.4]
    \draw[fill=black](0,0)--(1,0)--(1,1)--(0,1)--(0,0);
    \draw[fill=black](1,0)--(2,0)--(2,1)--(1,1)--(1,0);
    \draw(2,0)--(3,0)--(3,1)--(2,1)--(2,0);
    \draw[fill=black](3,0)--(4,0)--(4,1)--(3,1)--(3,0);
    \draw(4,0)--(5,0)--(5,1)--(4,1)--(4,0);
    \draw(5,0)--(6,0)--(6,1)--(5,1)--(5,0);
    \draw[fill=black](6,0)--(7,0)--(7,1)--(6,1)--(6,0);
    \draw[fill=white,color=white](0,-0.5)--(6,-0.5);

    \draw[fill=black](1,1)--(2,1)--(2,2)--(1,2)--cycle;
    \draw[fill=black](2,1)--(3,1)--(3,2)--(2,2)--cycle;
    \draw[fill=black](3,1)--(4,1)--(4,2)--(3,2)--cycle;
    \draw(4,1)--(5,1)--(5,2)--(4,2)--cycle;
    \draw[fill=black](5,1)--(6,1)--(6,2)--(5,2)--cycle;
    \draw(2,2)--(3,2)--(3,3)--(2,3)--cycle;
    \draw[fill=black](3,2)--(4,2)--(4,3)--(3,3)--cycle;
    \draw[fill=black](4,2)--(5,2)--(5,3)--(4,3)--cycle;
    \draw[fill=black](3,3)--(4,3)--(4,4)--(3,4)--cycle;
    \draw[fill=white](0,-1.485);
\end{tikzpicture}
\label{fig:pred.spacetime}
}
\hspace{30pt}
\subfigure[The commmunication interpretation of
$\predp{F_{110}}$.]
{
\begin{tikzpicture}[scale=0.4]

    \draw(-0.25,-0.7)--(-0.25,-0.8)--(2.75,-0.8)--(2.75,-0.7) (1.25,-1.4) node {Alice};

    \draw[fill=black](-0.25,-0.5)--(0.75,-0.5)--(0.75,0.5)--(-0.25,0.5)--cycle;
    \draw[fill=black](0.75,-0.5)--(1.75,-0.5)--(1.75,0.5)--(0.75,0.5)--cycle;
    \draw(1.75,-0.5)--(2.75,-0.5)--(2.75,0.5)--(1.75,0.5)--cycle;
    \draw(3.25,-0.7)--(3.25,-0.8)--(7.25,-0.8)--(7.25,-0.7) (5.25,-1.4) node {Bob};
    \draw[fill=black](3.25,-0.5)--(4.25,-0.5)--(4.25,0.5)--(3.25,0.5)--cycle;
    \draw(4.25,-0.5)--(5.25,-0.5)--(5.25,0.5)--(4.25,0.5)--cycle;
    \draw(5.25,-0.5)--(6.25,-0.5)--(6.25,0.5)--(5.25,0.5)--cycle;
    \draw[fill=black](6.25,-0.5)--(7.25,-0.5)--(7.25,0.5)--(6.25,0.5)--cycle;

    \draw[fill=black](1,1)--(2,1)--(2,2)--(1,2)--cycle;
    \draw[fill=black](2,1)--(3,1)--(3,2)--(2,2)--cycle;
    \draw[fill=black](3,1)--(4,1)--(4,2)--(3,2)--cycle;
    \draw(4,1)--(5,1)--(5,2)--(4,2)--cycle;
    \draw[fill=black](5,1)--(6,1)--(6,2)--(5,2)--cycle;
    \draw(2,2)--(3,2)--(3,3)--(2,3)--cycle;
    \draw[fill=black](3,2)--(4,2)--(4,3)--(3,3)--cycle;
    \draw[fill=black](4,2)--(5,2)--(5,3)--(4,3)--cycle;

    \draw(4.2,3.5)--(4.3,3.5)--(4.3,4.5)--(4.2,4.5) (5.75,4) node {Result};
    \draw[fill=black](3,3.5)--(4,3.5)--(4,4.5)--(3,4.5)--cycle;

\end{tikzpicture}
\label{fig:pred.cc}
}
\end{center}
\caption{Problem \predp{F_{110}}.}
\label{fig:pred}
\end{figure}

More precisely, for every $1 \leq i \leq (n-1)$, $\predp{F}|_n^i:Q^{i}
\times Q^{n-i}\rightarrow Q$ is such that
$\predp{F}|_n^i(x,y)=(f^\ast(xy))_1$.  This function $\predp{F}|_n^i$ can
be represented as a $|Q|^i \times |Q|^{n-i}$ matrix.  In other words,
we give $i$ states to Alice (rows) and $n-i$ states to Bob (columns);
i.e. $X=Q^i$ and $Y=Q^{n-i}$. We denote by $M_{F}^{n,i}$ such a
matrix. In the examples of Figure~\ref{fig178}, we have $n=2i+1=13$
and $n=2i+1=15$ (for the elementary CA Rule 178).

\begin{figure}[ht]
\centering
\fbox{\includegraphics[scale=0.3]{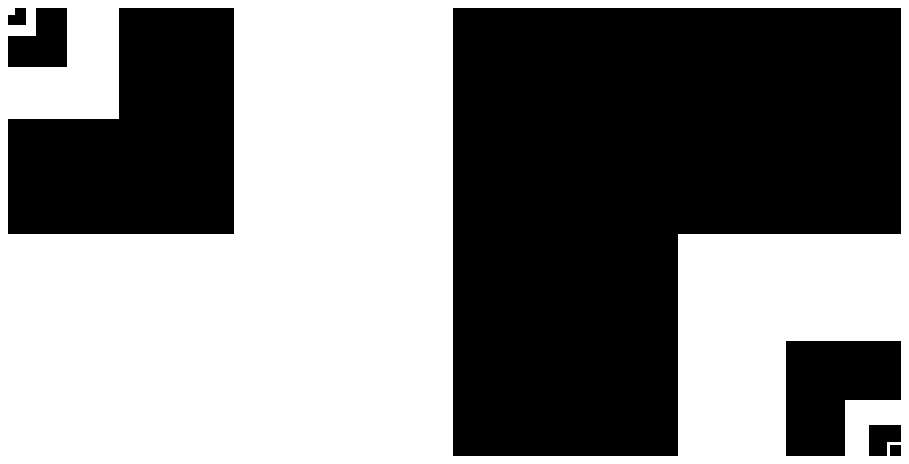}}\hspace*{30pt}
\fbox{\includegraphics[scale=0.3]{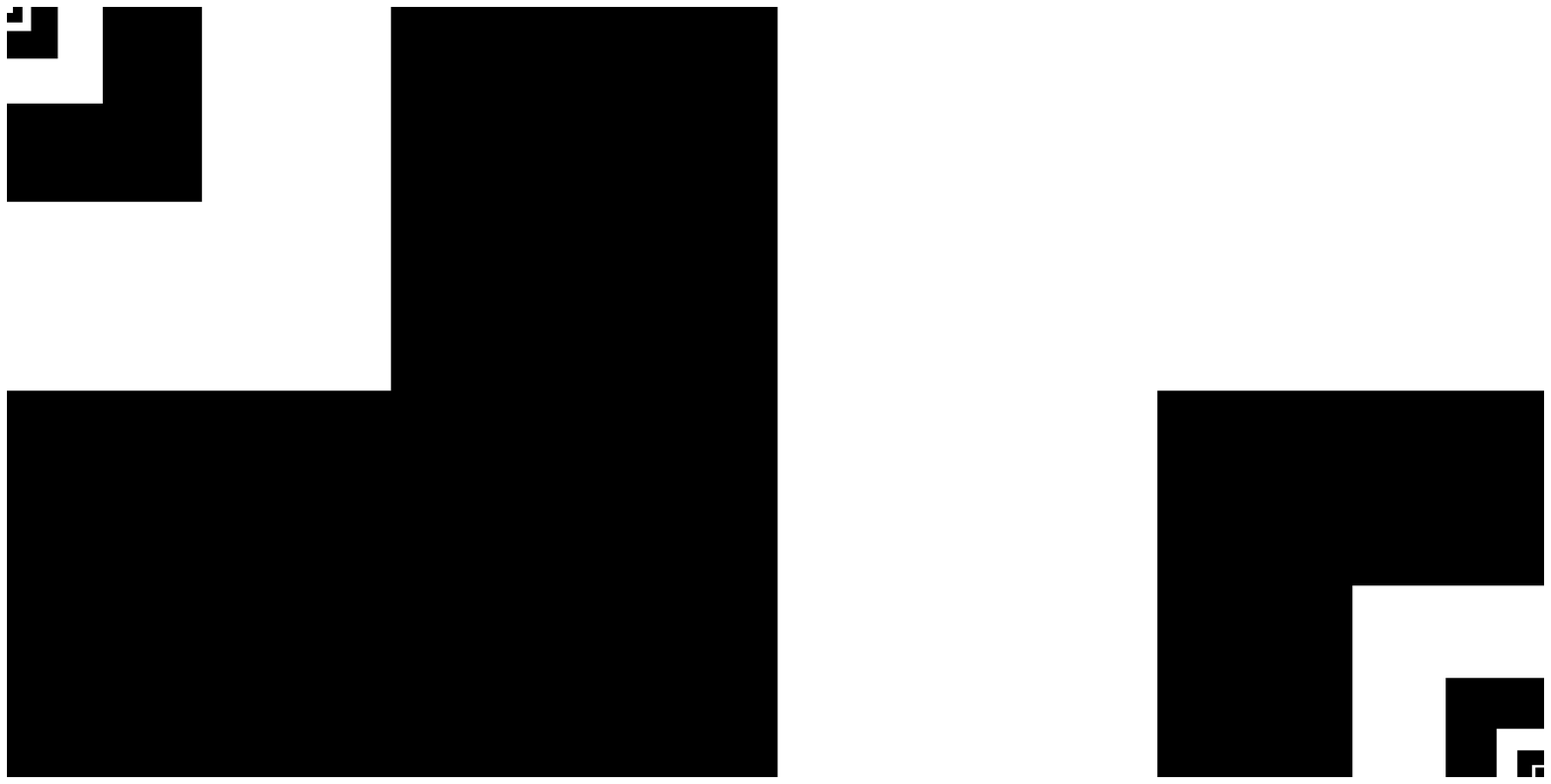}}
\caption{Matrices $M_{F_{178}}^{13,6}$ and $M_{F_{178}}^{15,7}$,
where ``178'' stands for the elementary CA Rule 178.}
\label{fig178}
\end{figure}

\begin{remark}
  We can consider the more restricted {\emph{one-round communication
      complexity}}.  In this setting only one party (either
  Alice or Bob) is allowed to send information. This restriction is
  justified by the fact that, according to a theorem of
  \cite{kushilevitz97}, by simply counting the number of different
  rows or columns of a certain matrix we obtain the {\emph{exact}}
  one-round communication complexity of the function. In our framework,
  the one round communication complexity of $\predp{F}|_n^i$
  corresponds to the minimum between the number of different rows and
  different columns of $M_{F}^{n,i}$.  Therefore, performing
  computational experiments in order to infer the one-round
  communication complexity of $\predp{F}|_n^i$, becomes an easy
  task.
\end{remark}

Recall that, given a CA $F$,  the communication complexity of \predp{F}
is defined as:

$$\cc{\predp{F}} = n\mapsto \max_{1\leq i\leq n-1} \CC\bigl({\predp{F}}|_n^i\bigr).$$

\begin{remark}
  In the above definition of $\predp{F}$, we choose a canonical local
  representation $(f,r)$ for the CA $F$. Replacing $f$ by
  another valid local representation can change the problem and its
  communication complexity. However this change would only introduce a
  multiplicative factor and therefore would not alter the main point
  of this paper (Section~\ref{sec:ncu}).
\end{remark}

Now we show that some well-known properties of CA induce small upper
bounds for the communication complexity of the prediction problem.
The results below are adaptations of ideas of~\cite{durr04} to the
formalism adopted in the present paper.

\begin{proposition}
  Let $F$ be any CA and $(f,r)$ be its canonical local representation.
  If there is a function ${g:\setN\rightarrow\setN}$ such that
  $f^{n}$ depends on only $g(n)$ cells, then $\cc{\predp{F}}
  \leq g(n)/2$.
\end{proposition}

In the work of M.~Sablik~\cite{sablik08}, CA which are equicontinuous
in some direction are considered. Following Theorem 4.3
of~\cite{sablik08}, they have a bounded number of dependant cells
(i.e, a bounded function $g(n)$). A well known example of such CA are
the nilpotent CA (a CA is \emph{nilpotent} if it converges to a unique
configuration from any initial configuration, or equivalently, if
$F^t$ is a constant function for any large enough $t$).

\begin{corollary}
 If $F$ is an equicontinuous CA in some direction then
\[\cc{\predp{F}} \in O(1).\]
\end{corollary}

Another set of CA with that property is the set of linear CA.  
A CA $F$ with state set $S$ is
linear if there is an operator $\oplus$ such that ${(S,\oplus)}$ is a
semi-group with neutral element $e$ and for all configurations $c$ and
$c'$ we have:
\[F(c\overoplus c') = F(c)\overoplus F(c'),\] where
$\overoplus$ denotes the uniform (cell-by-cell) extension of $\oplus$.

\begin{proposition}
If $F$ is a linear CA then $\cc{\predp{F}} \in O(1)$
\end{proposition}

\begin{proof}
The proof appears in~\cite{durr04} in a different setting. The idea
is that there is a simple one-round protocol to compute linear
functions: Alice and Bob can each compute on their own the image the
function would produce assuming the other party has only the neutral
element as input; then Alice or Bob communicate this result to the
other who can answer the final result by linearity.
\end{proof}

\subsection{Invasion}
\label{sec:inva}
Let $F$ be a CA and let $u$ be a given word.  Roughly, the problem
\invap{F}{u} is defined as follows: Given an input word $w$, we
define the $u$-periodic configuration $p_u$ on the one hand, and the
configuration $p_u(w)$ obtained by putting the word $w$ at the origin
over $p_u$ on the other hand; the invasion problem consists in determining
whether the differences between $p_u$ and $p_u(w)$ will expand to an
infinite width as time tends to infinity (hatched surface on
Figure~\ref{fig:inva}).


As we show in Proposition \ref{subsubsec:undecidable}, the general
case is, from the point of view of classical algorithmic theory,
undecidable.

\begin{figure}[ht]
\begin{center}
\begin{tikzpicture}[scale=0.5]
\draw[dashed](1,0.5)--(3,0.5);
\draw(3.25,0)--(6.25,0)--(6.25,1)--(3.25,1)--cycle (4.75,0.5) node {u};
\draw(6.5,0)--(9.5,0)--(9.5,1)--(6.5,1)--cycle (8,0.5) node {u};
\draw[fill=gray!20!white](9.75,0)--(15.75,0)--(15.75,1)--(9.75,1)--cycle (11.15,0.5) node {Alice}
    (14.1,0.5) node {Bob};
\draw[dashed] (12.5,0)--(12.5,1);
\draw (12.5, -0.5) node {$w$};
\draw(16,0)--(18.25,0)--(18.25,1)--(16,1) (16.625,0.5) node {u};
\draw(18.5,0)--(21.5,0)--(21.5,1)--(18.5,1)--cycle (20,0.5) node {u};
\draw[dashed](21.75,0.5)--(23.75,0.5);
\draw[white](5.75,5)--(19.75,5);
\path[draw=black,pattern=big north east lines,decorate,decoration=random steps](9.75,1)--(5.75,5)
    --(19.75,5)--(15.75,1);
\draw[draw=white,fill=white,decorate,decoration=random steps](5,4.5) rectangle (20,5.5);
\end{tikzpicture}
\end{center}
\caption{The \textsc{invasion} problem}
\label{fig:inva}

\end{figure}
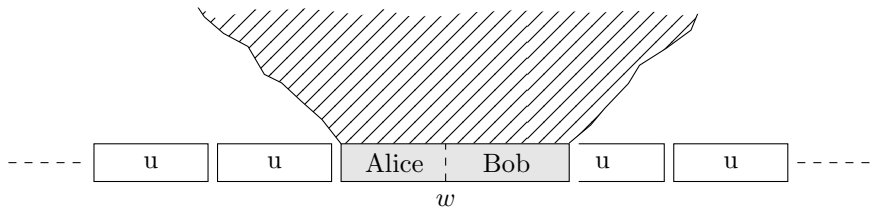

Now we give formal definitions.

\begin{definition}
  \label{def:invasion}
  Let $u=u_1 \ldots u_l$ be a finite word.  Let $p_u$ be such that for
  all $i\in\setZ$, $p_u[i]=u[i\mod l]$.

 \begin{itemize}
  \item we consider the ultimately periodic orbit
    ${\bigl(F^t(\perio{u})\bigr)_t}$ as the reference orbit;
  \item for each ${x_1,\ldots,x_n\in Q}$, we define the configuration
    $\perio{u}(x_1,\ldots,x_n)$ obtained by modifying
    $\perio{u}$ as follows:
    \[\perio{u}(x_1,\ldots,x_n)_z =
    \begin{cases}
      (\perio{u})_z&\text{for }z\leq 0\text{ or }z\geq n+1,\\
      x_z&\text{otherwise.}
    \end{cases}\]
  \item for each $t$, we denote $\lediff{t}$ and $\ridiff{t}$ the
    lefmost and rightmost differences between the $t^\text{th}$ images
    of $\perio{u}$ and $\perio{u}(x_1,\ldots,x_n)$: 
    \begin{align*}
      \lediff{t} &= \min\bigl\{z : F^t(\perio{u})_z\not= F^t\bigl(\perio{u}(x_1,\ldots,x_n)\bigr)_z\bigr\},\\
      \ridiff{t} &= \max\bigl\{z : F^t(\perio{u})_z\not= F^t\bigl(\perio{u}(x_1,\ldots,x_n)\bigr)_z\bigr\}.
    \end{align*}
  \item then $\invap{F}{u}(x_1 \ldots x_n)$ equals $1$ if
    ${\ridiff{t}-\lediff{t} \rightarrow_t \infty}$ and $0$ otherwise.
  \end{itemize}
\end{definition}

As explained before, we associate to any $F$ and $u$, the communication
complexity of $\invap{F}{u}$ defined as $\cc{\invap{F}{u}}$.





Some CA have by nature a trivial invasion complexity because their
dynamics consists in propagating errors systematically. This is the
case of (positively) expansive CA. Recall that $F$ is (positively)
expansive if there is some $\epsilon>0$ such that:
\[\forall x,y,\ x\not=y\Rightarrow \exists t, d(F^t(x),F^t(y))\geq\epsilon\]
where $d$ is the Cantor distance.

\begin{proposition}
  Let $F$ be a positively expansive CA. Then for all $u$ we have
  ${\cc{\invap{F}{u}} = 1}$.
\end{proposition}

\begin{proof}
  Fix any $u$ and consider any ${(x_1,\ldots,x_n)}$ such that
  ${\perio{u}(x_1,\ldots,x_n)\not=\perio{u}}$. By classical results of
  P.~K\r urka~\cite{kurka97}, there is a positive constant $\alpha$
  (average propagation speed) such that ${\lediff{t}} \leq -\alpha t$
  and ${\ridiff{t}} \geq \alpha t$. Therefore, invasion occurs if and
  only if:
  \[\perio{u}(x_1,\ldots,x_n)\not=\perio{u}.\]
  
  Testing this condition can be done with only $1$ bit of
  communication: Either Alice or Bob communicates whether she (or he)
  sees any difference between her (or his) input and the corresponding part of
  $p_u$; then the other party  can answer. The proposition follows.
\end{proof}



\subsection{Cycle length}
\label{sec:cycl}

For this last problem, we consider spatially periodic configurations.
Since there are only a finite number of such configurations of a given
period size, and the size of the period does not grow with time, then
clearly the evolution becomes periodic (in time) after a certain
number of steps (see Figure~\ref{fig:cycle-length} where successive
steps are represented by successive concentric circles). Roughly
speaking, the \textsc{cycle} problem consists in determining whether the
length of this ultimate (temporal) period is small, starting from a
given (spatially) periodic initial configuration. The formal definition
follows.

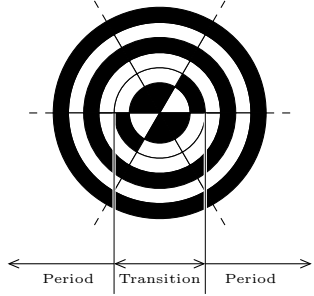
\begin{figure}[ht]
\begin{center}
\begin{tikzpicture}[scale=0.2]
\draw[fill=white](0.0:0)--(0.0:2) arc (0.0:60.0:2)--(60.0:0) arc (60.0:0.0:0);
\draw[fill=black](60.0:0)--(60.0:2) arc (60.0:120.0:2)--(120.0:0) arc (120.0:60.0:0);
\draw[fill=black](120.0:0)--(120.0:2) arc (120.0:180.0:2)--(180.0:0) arc (180.0:120.0:0);
\draw[fill=white](180.0:0)--(180.0:2) arc (180.0:240.0:2)--(240.0:0) arc (240.0:180.0:0);
\draw[fill=black](240.0:0)--(240.0:2) arc (240.0:300.0:2)--(300.0:0) arc (300.0:240.0:0);
\draw[fill=black](300.0:0)--(300.0:2) arc (300.0:360.0:2)--(360.0:0) arc (360.0:300.0:0);
\draw[fill=black](0.0:2)--(0.0:3) arc (0.0:60.0:3)--(60.0:2) arc (60.0:0.0:2);
\draw[fill=white](60.0:2)--(60.0:3) arc (60.0:120.0:3)--(120.0:2) arc (120.0:60.0:2);
\draw[fill=white](120.0:2)--(120.0:3) arc (120.0:180.0:3)--(180.0:2) arc (180.0:120.0:2);
\draw[fill=black](180.0:2)--(180.0:3) arc (180.0:240.0:3)--(240.0:2) arc (240.0:180.0:2);
\draw[fill=white](240.0:2)--(240.0:3) arc (240.0:300.0:3)--(300.0:2) arc (300.0:240.0:2);
\draw[fill=white](300.0:2)--(300.0:3) arc (300.0:360.0:3)--(360.0:2) arc (360.0:300.0:2);
\draw[fill=white](0.0:3)--(0.0:4) arc (0.0:60.0:4)--(60.0:3) arc (60.0:0.0:3);
\draw[fill=white](60.0:3)--(60.0:4) arc (60.0:120.0:4)--(120.0:3) arc (120.0:60.0:3);
\draw[fill=white](120.0:3)--(120.0:4) arc (120.0:180.0:4)--(180.0:3) arc (180.0:120.0:3);
\draw[fill=white](180.0:3)--(180.0:4) arc (180.0:240.0:4)--(240.0:3) arc (240.0:180.0:3);
\draw[fill=white](240.0:3)--(240.0:4) arc (240.0:300.0:4)--(300.0:3) arc (300.0:240.0:3);
\draw[fill=white](300.0:3)--(300.0:4) arc (300.0:360.0:4)--(360.0:3) arc (360.0:300.0:3);
\draw[fill=black](0.0:4)--(0.0:5) arc (0.0:60.0:5)--(60.0:4) arc (60.0:0.0:4);
\draw[fill=black](60.0:4)--(60.0:5) arc (60.0:120.0:5)--(120.0:4) arc (120.0:60.0:4);
\draw[fill=black](120.0:4)--(120.0:5) arc (120.0:180.0:5)--(180.0:4) arc (180.0:120.0:4);
\draw[fill=black](180.0:4)--(180.0:5) arc (180.0:240.0:5)--(240.0:4) arc (240.0:180.0:4);
\draw[fill=black](240.0:4)--(240.0:5) arc (240.0:300.0:5)--(300.0:4) arc (300.0:240.0:4);
\draw[fill=black](300.0:4)--(300.0:5) arc (300.0:360.0:5)--(360.0:4) arc (360.0:300.0:4);
\draw[fill=white](0.0:5)--(0.0:6) arc (0.0:60.0:6)--(60.0:5) arc (60.0:0.0:5);
\draw[fill=white](60.0:5)--(60.0:6) arc (60.0:120.0:6)--(120.0:5) arc (120.0:60.0:5);
\draw[fill=white](120.0:5)--(120.0:6) arc (120.0:180.0:6)--(180.0:5) arc (180.0:120.0:5);
\draw[fill=white](180.0:5)--(180.0:6) arc (180.0:240.0:6)--(240.0:5) arc (240.0:180.0:5);
\draw[fill=white](240.0:5)--(240.0:6) arc (240.0:300.0:6)--(300.0:5) arc (300.0:240.0:5);
\draw[fill=white](300.0:5)--(300.0:6) arc (300.0:360.0:6)--(360.0:5) arc (360.0:300.0:5);
\draw[fill=black](0.0:6)--(0.0:7) arc (0.0:60.0:7)--(60.0:6) arc (60.0:0.0:6);
\draw[fill=black](60.0:6)--(60.0:7) arc (60.0:120.0:7)--(120.0:6) arc (120.0:60.0:6);
\draw[fill=black](120.0:6)--(120.0:7) arc (120.0:180.0:7)--(180.0:6) arc (180.0:120.0:6);
\draw[fill=black](180.0:6)--(180.0:7) arc (180.0:240.0:7)--(240.0:6) arc (240.0:180.0:6);
\draw[fill=black](240.0:6)--(240.0:7) arc (240.0:300.0:7)--(300.0:6) arc (300.0:240.0:6);
\draw[fill=black](300.0:6)--(300.0:7) arc (300.0:360.0:7)--(360.0:6) arc (360.0:300.0:6);
\draw[white,very thick](3,0)--(3,-12) (-3,0)--(-3,-12);
\draw[black,very thin](3,0)--(3,-12) (-3,0)--(-3,-12);
\draw[black](-4,0)--(4,0);
\draw[angle 45-angle 45](-3,-10)--(3,-10);
\draw[-angle 45](3,-10)--(10,-10);
\draw[angle 45-](-10,-10)--(-3,-10);
\draw(-6,-11)node{\tiny Period} 
(0,-11)node{\tiny Transition}
(6,-11)node{\tiny Period};
{[draw=black]
\foreach \x in {0,60,120,180,240,300}{
\draw[dashed](\x:7)--(\x:9);}
}
\end{tikzpicture}

\end{center}
\caption{The \cyclp{}{} problem on elementary CA Rule 33. Since the configurations
are cyclic, we can represent one configuration on a circle. Time goes from
the inner circle to outer circles,
zeros are white, and ones are black. For instance, the
initial configuration -- on the innermost circle -- is 011011. After one step, it
becomes 100100.}
\label{fig:cycle-length}
\end{figure}

\begin{definition}
  Let $F$ be a CA and let $k\geq1$.  For any $u\in Q^\ast$ we denote
  by $\lambda(u)$ the length of the ultimate period of the orbit of
  configuration $\perio{u}$ under $F$:
  \[\lambda(u) = \min \bigl\{p : \exists t_0, \forall t\geq t_0,
  F^t(p_u) = F^{t+p}(p_u)\bigr\}.\] 
  The problem $\cyclp{F}{k}$ is then defined by:
  \[\cyclp{F}{k}(u) =
  \begin{cases}
    1 &\text{if }\lambda(u)\leq k,\\
    0 &\text{otherwise.}
  \end{cases}\]
\end{definition}

One of the interests of the cycle length problem lies in the following
complexity upper bound for reversible CA.

\begin{proposition}
  \label{prop:reversible}
  Let $F$ be any reversible CA. Then, for any fixed $k$, we have:
  \[\cc{\cyclp{F}{k}}\in O(1).\]
\end{proposition}
\begin{proof}
  For a reversible CA, orbits of periodic configurations are not only
  ultimately periodic but also periodic. More precisely, for any
  periodic configuration $c$, the cycle length starting from $c$ is
  less than $k$ if and only if:
  \[\exists t\leq k: F^t(c)=c.\]

  Thus, Alice and Bob can simply simulate the automaton for $k$ steps,
  then check if a configuration repeats during these steps : this can be done with
  $4k\cdot r\cdot \lfloor 1+\log Q\rfloor$ bits, to transmit the cells next to the
  border between Alice and Bob's respective parts, then one bit for Alice to
  tell Bob if a configuration appeared twice during the $k$ steps.
  
\end{proof}

\section{The three corresponding necessary conditions for 
intrinsic universality}
\label{sec:univ}

In this section we show that intrinsic universality implies that the
communication complexity of the three canonical problems described above
must be maximal. Before giving precise definitions, recall that a CA is
intrinsically universal if it is able to simulate any other CA. Our approach
with communication complexity proceeds in two steps:

\begin{itemize}
\item we show that the simulation of $F$ by $G$ implies a reduction
  from any canonical problem for $F$ to the corresponding problem for
  $G$ in such a way that the communication complexity is preserved (up to
  some distortions involving only multiplicative factors);

\item we show the existence of maximal communication
  complexity CA for each of the canonical problems.
\end{itemize}

Before developing these two steps, we give formal definitions for
simulations and intrinsic universality.

\subsection{Simulations and universality}

The base ingredient is the relation of sub-automaton.  A CA $F$ is a
\emph{sub-automaton} of a CA $G$, denoted by ${F\sac G}$, if there is
an injective map $\iota$ from $\alphA$ to $\alphB$ such that
${\overline{\iota}\circ F=G\circ \overline{\iota}}$, where
${\overline{\iota}:\alphA^\ZZ\rightarrow \alphB^\ZZ}$ denotes the
uniform extension of $\iota$.

A CA $F$ simulates a CA $G$ if some \emph{rescaling} of $G$
is a sub-automaton of some \emph{rescaling} of $F$. The ingredients
of the rescalings are simple: packing cells into blocks, iterating the
rule and composing with a translation. Formally, given any state set
$Q$ and any $m\geq 1$, we define the bijective packing map ${\bloc{m}:
  Q^\ZZ\rightarrow \bigl(Q^m\bigr)^\ZZ}$ by:
\[\forall z\in\ZZ : \bigl(\bloc{m}(c)\bigr)(z) = \bigl(c(mz),\ldots,c(mz+m-1)\bigr)\]
for all ${c\in Q^\ZZ}$. The rescaling $\grp{F}{m,t,z}$ of $F$ by
parameters $m$ (packing), ${t\geq 1}$ (iterating) and ${z\in\ZZ}$
(shifting) is the CA of state set $Q^m$ and global rule:
\[\bloc{m} \circ \sigma_z\circ F^t \circ \debloc{m}.\]
The fact that the above function is the global rule of a cellular
automaton follows from Curtis-Lyndon-Hedlund theorem \cite{hedlund69}
because it is continuous and commutes with translations. With these
definitions, we say that $G$ simulates $F$, denoted ${F\simu G}$, if
there are rescaling parameters $m_1$, $m_2$, $t_1$, $t_2$, $z_1$ and
$z_2$ such that ${\grp{F}{m_1,t_1,z_1}\sac\grp{G}{m_2,t_2,z_2}}$.

We can now naturally define the notion of universality associated to
this simulation relation.

\begin{definition}
  $F$ is \emph{intrinsically universal} if for all $G$ it holds that
  ${G\simu F}$. $F$ is \emph{reversible universal} if for all
  \emph{reversible} $G$ it holds that ${G\simu F}$.
\end{definition}

We consider the following relation of comparison
between functions from $\setN$ to $\setN$:
\[\phi_1\moindre \phi_2\iff\exists \alpha,\beta,\gamma,\delta \geq 1,\forall n\in\setN:
\phi_1(\alpha n)\leq\beta\phi_2(\gamma n)+\delta.\]

\begin{remark}
  All the functions we will compare by $\moindre$ are in $O(n)$ since
  they come from a communication complexity problem. Moreover, the set
  of such functions that are in $\Omega(n)$ form an equivalence class
  for $\moindre$. Although we sometimes give more precise bounds, most
  of the paper focuses on whether or not some function belongs to this class.
\end{remark}

\begin{proposition}
  \label{prop:simpred}
  If $F\simu G$ then ${\cc{\predp{F}}\moindre\cc{\predp{G}}}$.
\end{proposition}
\begin{proof}
  We successively consider each ``ingredient'' involved in the simulation
  relation. 
  \begin{description} 
  \item[Sub-automaton: ] if ${F\sac G}$ then each valid protocol to
    compute $\predp{G}|_n^i$ is also a valid protocol to compute
    iterations of $\predp{F}|_n^i$ (up to state
    renaming).
  \item[Iterating: ]
    We have ${\cc{\predp{F^t}} \in\Theta(\cc{\predp{F}})}$.
    In fact, if we have a protocol for the prediction problem of $F^t$ -- which
    is an automaton of radius $t\cdot r$ -- then we can use it to predict $F$ :
    on a configuration $x$ of size $n$, we use the protocol to predict the result
    of iterating $\lfloor\frac{n}{r\cdot t}\rfloor$ times $F^t$,
    which gives a configuration of size at most $t\cdot r-1$. To do this, we
    just use the protocol at most $t\cdot r-1$ times to predict each cell
    of this configuration, then Alice or Bob conclude by simulating the
    automaton directly.
    
    The other direction is even simpler: a protocol for $\predp{F}$ can be used
    directly for $\predp{F^t}$ by just slightly reducing the input of Bob.
  \item[Shifting: ] This operation only affects the splitting of
    inputs. Since we always take in each case the splitting of maximum
    complexity, this has no influence on the final complexity
    function.  
  \item[Packing: ] let $F$ be any CA and $n$ be fixed.  Consider the
    problem $\predp{\grp{F}{m,1,0}}|_n^j$ for some $j$. Now consider
    any sequence of valid protocols $(P_i)$, one for each problem
    ${\predp{F}|_{nm}^i}$. It follows from the the definition of
    packing maps that $\predp{\grp{F}{m,1,0}}|_n^j$ can be solved by
    applying $m$ suitably chosen protocols in the sequence
    $(P_i)$. Therefore \[\cc{\predp{\grp{F}{m,1,0}}}(n)\leq
      m\cdot\cc{\predp{F}}(n)\]

    Reciprocally, one has for all $n$: \[\cc{\predp{F}}(n)\leq
      \cc{\predp{\grp{f}{m,1,0}}}({\lceil n/m\rceil})+m\] where the
    additional constant $m$ is used to deal with input splittings of
    $\predp{F}|_n$ which have no equivalent in
    $\predp{\grp{f}{m,1,0}}|_{\lceil n/m\rceil}$ because they do not
    cut the input at a position which is multiple of $m$.
  \end{description} 
  Therefore we have: $\cc{\predp{F}}\moindre\predp{\grp{F}{m,t,z}}$,
  $\predp{\grp{F}{m,t,z}}\moindre\cc{\predp{F}}$ and if ${F\sac G}$
  then ${\cc{\predp{F}}\moindre\predp{G}}$. The proposition follows.
\end{proof}

The following result shows that the invasion complexity is increasing
with respect to simulations.

\begin{proposition}
  \label{prop:siminva}
  If $F\simu G$ then for all $u$ there is $v$ such that
  \[\cc{\invap{F}{u}}\moindre\cc{\invap{G}{v}}.\]
\end{proposition}
\begin{proof}
  The simulation relation $\simu$ is such that ultimately periodic
  configurations of $F$ are converted into ultimately periodic
  configurations of $G$. Hence, the invasion problem of $F$ reduces to
  the invasion problem of $G$. More precisely, it is sufficient to
  check the following properties, each dealing with an aspect of the
  simulation relation $\simu$:
  \begin{itemize}
  \item for any CA $F$, any $u$ and any rescaling parameters $m,t,z$,
    we have
    \[\cc{\invap{F}{u}}\moindre\cc{\invap{\grp{F}{m,t,z}}{U}}\] where
    $U$ is the period of the configuration $\bloc{m}(\perio{u})$;
  \item if ${F\sac G}$ then, for any $u$, ${\cc{\invap{F}{u}}\moindre\cc{\invap{G}{u}}}$;
  \item for any CA $F$, any rescaling parameters $m,t,z$, any $U$
    (over the alphabet of $\grp{F}{m,t,z}$)
    ${\cc{\invap{\grp{F}{m,t,z}}{U}}\moindre\cc{\invap{F}{u}}}$ where
    $u$ is the period of the configuration $\debloc{m}(\perio{U})$.
  \end{itemize}
  The result follows by composition of the $3$ properties above.
\end{proof}

Finally, we show a similar result for the cycle length problem. The
problem is parametrized by an integer $k$ and the following
proposition establishes that for suitable but arbitrary large values
of this parameter the complexity of the problem is
conserved.

\begin{proposition}
  \label{prop:simcycl}
  If $F\simu G$ then for all $k_0$ there is $k$ and $k'$ such that:
  \begin{itemize}
  \item $k\geq k_0$ and $k'\geq k_0$;
  \item ${\cc{\cyclp{F}{k}}\moindre\cc{\cyclp{G}{k'}}}$.
  \end{itemize}
\end{proposition}
\begin{proof}
  The effect of rescaling transformations on cyclic orbits of periodic
  configurations is to change the (spatial) period length as well as
  the (temporal) cycle length. More precisely, we have:
  \begin{itemize}
  \item if ${F\sac G}$ then, for any $k$, ${\cc{\cyclp{F}{k}}\moindre\cc{\cyclp{G}{k}}}$;
  \item for any $k$,
    \begin{itemize}
    \item ${\cc{\cyclp{F}{k}}\moindre\cc{\cyclp{\grp{F}{m,1,0}}{k}}}$ and
    \item ${\cc{\cyclp{\grp{F}{m,1,0}}{k}}\moindre\cc{\cyclp{F}{k}}}$;
    \end{itemize}
  \item for any $t$ and any $k$ we have:
    \[\cc{\cyclp{\grp{F}{1,t,0}}{k}}\moindre\cc{\cyclp{F}{kt}};\]
  \item for any $t$ and any $k$ such that ${k\bmod  t = 0}$ we have:
    \[\cc{\cyclp{F}{k}}\moindre\cc{\cyclp{\grp{F}{1,t,0}}{k/t}}.\]
  \end{itemize}
  The proposition follows.
\end{proof}

\subsection{Existence of CA with maximal complexity}

This section is devoted to the following existence result.

\begin{proposition}\ 
  \label{exists}

  \begin{enumerate}
  \item There exists a reversible CA $F$ and a word $u$ with $\cc{\invap{F}{u}}\in\Omega(n)$.
  \item There exists a reversible CA $F$ with $\cc{\predp{F}}\in\Omega(n)$.
  \item There exists a CA $F$ s.t. for any ${k\geq 1}$, $\cc{\cyclp{F}{k}}\in\Omega(n)$.
  \end{enumerate}
\end{proposition}

We now define the reversible CA of assertion $2$ of
Proposition~\ref{exists}, which we call $G$ in the sequel. It is made of $3$
layers:
\begin{itemize}
\item flag layer $Q_f=\{0,1\}$,
\item circulation layer $Q_c=\{W\}\cup\{0,1\}\times\{0,1\}$,
\item test layer $Q_t=\{0,1\}\times\{0,1\}$.
\end{itemize}

The flag layer is simply the identity over $Q_f$. The circulation
layer does not depend on other layers and has the following behaviour.
\begin{figure}[htp]
  \centering

  \begin{tikzpicture}[scale=0.38]
    \draw[draw=none,fill=black!3!white](-13,-1)rectangle(13,4);

    \draw[dotted](-9,3.5) node[anchor=west,fill=black!3!white] {$x_1$}
    --(-0.5,3.5) node[anchor=east, fill=black!3!white]{$x_n$};
    \draw[dotted](0.5,3.5) node[anchor=west,fill=black!3!white] {$0$}
    --(9,3.5) node[anchor=east, fill=black!3!white]{$0$};
    \draw (0,3.5) node {$0$};

    \draw[dotted](-9,2.5) node[anchor=west,fill=black!3!white] {$0$}
    --(-0.5,2.5) node[anchor=east, fill=black!3!white]{$0$};
    \draw[dotted](0.5,2.5) node[anchor=west,fill=black!3!white] {$y_n$}
    --(9,2.5) node[anchor=east, fill=black!3!white]{$y_1$};
    \draw (0,2.5) node {$0$};

    \draw[dotted](-0.5,1.5)node[fill=black!3!white,anchor=east]{$0$}--(-5,1.5)
    (0.5,1.5)node[fill=black!3!white,anchor=west]{$0$}--(5,1.5);

    \draw (-0.5,1) rectangle (0.5,4)
    (0,1.5) node {$1$};
    \draw(-13,3)--(13,3);
    \foreach\i in {-11,10}
    \draw[draw=none,fill=black!20!white](\i,2)rectangle(\i+1,4)
    (\i+0.5,3) node {$W$};
    \foreach\i in {-1,0,1,2,4}
    \draw(-13,\i)--(13,\i);
    
    \draw(-13.2,4)--(-13.3,4)--(-13.3,2)--(-13.2,2);
    \draw(-13.3,3) node [anchor=east] {\small circulation};    
    \draw(-13.3,1.5)node[anchor=east]{\small flag};
    \draw(-13.2,-1)--(-13.3,-1)--(-13.3,1)--(-13.2,1);    
    \draw(-13.3,0)node[anchor=east]{\small test};
    
    \foreach\i in {-9,12}
    {
      \draw[-latex'] (\i,2.5)..controls (\i-1,2.5) and (\i-1,3.5)..(\i,3.5);
      \draw[-latex'] (-\i,3.5)..controls (-\i+1,3.5) and (-\i+1,2.5)..(-\i,2.5);
    }


    \begin{scope}[yshift=-7cm]
    \draw[draw=none,fill=black!3!white](-13,-1)rectangle(13,4);
    \draw[dotted](-5,3.5) node[anchor=west,fill=black!3!white] {$x_1$}
    --(3.5,3.5) node[anchor=east, fill=black!3!white]{$x_n$};
    \draw[dotted]
    (-9,3.5) node[anchor=west,fill=black!3!white]{$0$}--
    (-5,3.5) node[anchor=east,fill=black!3!white]{$0$}
    (3.5,3.5) node[anchor=west,fill=black!3!white] {$0$}
    --(9,3.5) node[anchor=east, fill=black!3!white]{$0$};
    \draw (0,3.5) node[fill=black!3!white] {$x_k$};

    \draw[dotted]
    (-9,2.5) node[anchor=west,fill=black!3!white]{$0$}--
    (-3.5,2.5) node[anchor=east,fill=black!3!white]{$0$}
    (5,2.5) node[anchor=west,fill=black!3!white] {$0$}
    --(9,2.5) node[anchor=east, fill=black!3!white]{$0$};
    \draw[dotted](-3.5,2.5) node[anchor=west,fill=black!3!white] {$y_n$}
    --(5,2.5) node[anchor=east, fill=black!3!white]{$y_1$};
    \draw (0,2.5)node[fill=black!3!white]{$y_k$};

    \draw[dotted](-0.5,1.5)node[fill=black!3!white,anchor=east]{$0$}--(-5,1.5)
    (0.5,1.5)node[fill=black!3!white,anchor=west]{$0$}--(5,1.5);

    \draw[latex'-](-8,0.5)--(-6,0.5);
    \draw[latex'-](8,-0.5)--(6,-0.5);
    \pgfmathdeclarerandomlist{state}{{$0$}{$1$}}
    \foreach\i in {2,3,4,5,6}
    {
      \pgfmathrandomitem{\s}{state}
      \draw(\i,-0.5)node[fill=black!3!white]{\s};
      \draw(-\i,0.5)node[fill=black!3!white]{\s};
    }
    \draw (-0.5,1) rectangle (0.5,4)
    (0,1.5) node {$1$};
    \draw(-13,3)--(13,3);
    \foreach\i in {-11,10}
    \draw[draw=none,fill=black!20!white](\i,2)rectangle(\i+1,4)
    (\i+0.5,3) node {$W$};
    \foreach\i in {-1,0,1,2,4}
    \draw(-13,\i)--(13,\i);

    \draw(-13.2,4)--(-13.3,4)--(-13.3,2)--(-13.2,2);
    \draw(-13.3,3) node [anchor=east] {\small circulation};
    \draw(-13.3,1.5)node[anchor=east]{\small flag};
    \draw(-13.2,-1)--(-13.3,-1)--(-13.3,1)--(-13.2,1);    
    \draw(-13.3,0)node[anchor=east]{\small test};
    
    \foreach\i in {-9,12}
    {
      \draw[-latex'] (\i,2.5)..controls (\i-1,2.5) and (\i-1,3.5)..(\i,3.5);
      \draw[-latex'] (-\i,3.5)..controls (-\i+1,3.5) and (-\i+1,2.5)..(-\i,2.5);
    }
\end{scope}
                                   
  \end{tikzpicture}
  \caption{Above: initial configuration. Below: the configuration $k$ steps later.}
\label{fig:circule}
\end{figure}
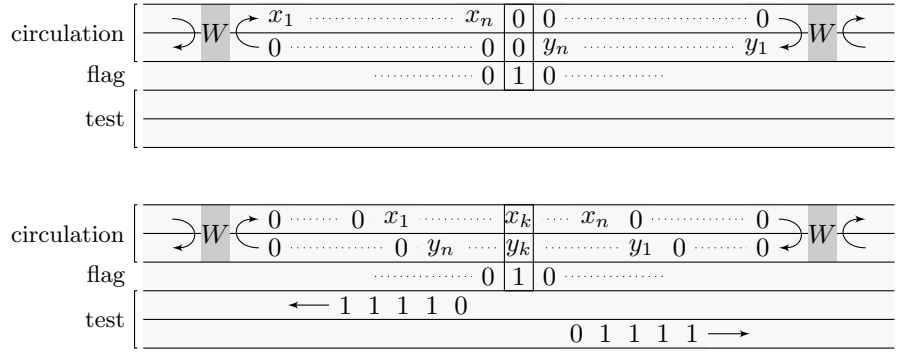
\begin{itemize}
\item normal states in $\{0,1\}\times\{0,1\}$ represent two sub-layers
  (top and bottom) and, if no $W$ state is in the neighbourhood, the
  top sub-layer simply shifts to the right and the bottom sub-layer simply
  shifts to the left.
\item $W$ states are walls: They stay unchanged forever. Moreover, a
  normal cell on the right of a wall has the following behaviour: The
  top value shifts to the right and the bottom value goes to the top. A
  normal cell on the left of a wall has a symmetric behaviour: The
  bottom value shifts to the  left and the top value goes to the bottom. See
  figure~\ref{fig:circule}.
\end{itemize}

Finally, the test layer is made of two sub layers (top and bottom)
wich are independant. The top layer does the following:
\begin{itemize}
\item if the flag layer of the cell is $1$ and if the circulation
  layer contains the state $(1,1)$ then invert bit and shift right;
\item in any other case, simply shift right.
\end{itemize}

The bottom sub-layer does the same but replace right by left.

\begin{proof}[Proof of Proposition~\ref{exists}]\ 
  
  \begin{enumerate}
  \item We first show that $G$ defined above has the properties of
    assertion 1 of the proposition. First, it is reversible: the flag and
    circulation layers are themselves reversible, and the knowledge of 
    these two layers makes the flag layer reversible too.

    Now let $q_0$ be the state where flag layer is $0$, circulation
    layer is $(0,0)$ and the test layer is $(0,0)$. Consider input bits
    $x_1,\ldots, x_n$ on the one hand and $y_1,\ldots, y_n$ on the other
    hand. Let $X_i$ be the state with flag layer $0$, test layer
    $(0,0)$ and circulation layer $(x_i,0)$. Similarily let $Y_i$ be
    the state with flag layer $0$, test layer $(0,0)$ and circulation
    layer $(0,y_i)$. Let $M$ be the state of flag layer $0$,
    circulation layer $W$ and test layer $(0,0)$. Finally let $T$ be
    the state of flag layer $1$, circulation layer $(0,0)$ and test
    layer $(0,0)$.  Consider the configuration
    $C(x_1,\ldots,x_n,y_1,\ldots,y_n)$:
    \[{}^\omega q_0\ M\ X_n\cdots X_1\ T\ Y_1\cdots Y_n\ M\
    q_0^\omega\] We can consider this configuration as an instance of
    the invasion problem \invap{F^{2n+3}}{u} where $u=q_0$. The only
    possible invasion in such an instance comes from the test
    layer. It follows from the definition of $G$ that there is
    invasion on this instance if and only if
    \[\exists i, x_i=y_i=1.\]
    Hence, the \textsc{disjointness} problem reduces to the invasion
    problem through such instances. Using
    proposition~\ref{prop:commlb}, we conclude that
    $\cc{\invap{G}{q_0}}\in\Omega(n)$.
  \item Assertion 2 of the proposition can be proven with a CA $F$
    simpler than $G$, but using similar ideas. $F$ has radius $1$ and
    its state set is the product of 3 components:
    \begin{itemize}
    \item left circulation with state set $\{0,1\}$,
    \item right circulation with state set $\{0,1\}$,
    \item test with state set $\{0,1\}$.
    \end{itemize}
    The behaviour is the following:
    \begin{itemize}
    \item each of the left and right circulation components are
      independent of the other components and consists in simple shift
      (left and right respectively),
    \item the test component simply flips its value if both left
      and right circulation components have value $1$ and stays
      unchanged else.
    \end{itemize}
    $F$ is clearly reversible (circulation layers are independent
    shifts and test layer is reversible knowing other components).
    Moreover, the inner product problem reduces to the prediction
    problem of $F$. Indeed, for any ${x,y\in\{0,1\}^n}$ consider the
    word
    \[w = X_1\cdots X_nZY_n\cdots Y_1\] where $X_i$ is the state equal
    to $x_i$ on the right circulation component and $0$ elsewhere,
    $Y_i$ is the state equal to $y_i$ on the left circulation
    component and $0$ else, and $Z$ is the state equal to $0$
    everywhere. It follows from definition of $F$ that
    \[\predp{F}|_n(w) = 1 \iff \sum x_iy_i \bmod 2 = 1.\]
    proposition~\ref{prop:commlb} implies that
    ${\cc{\predp{F}}\in\Omega(n)}$.
  \item 

We use the problem \DISJ to build a hard \cyclp{}{} problem. The idea
is that if Alice and Bob receive two disjoint sets as their inputs,
our CA will check \DISJ forever. Otherwise it will erase all the
tape, leaving a uniform, 1-periodic, configuration.

We use three layers in this construction, let us call the corresponding rules
$F_1$, $F_2$ and $F_3$. They are all of radius one, and all use the same
set of states $\{0,1,K\}$.
The $K$ state is used to erase all three tapes: thus, if it appears on
any component, it spreads on all three.

On (local) configurations not involving $K$, $F_1$ is a simple left shift,
and $F_2$ a simple right shift. We use $F_3$ as a control layer: we need to
check if the two other components represent two disjoint sets. The correponding
bitwise operation is: \[\bigwedge_{i=1}^n \neg(x_i\wedge y_i)\]
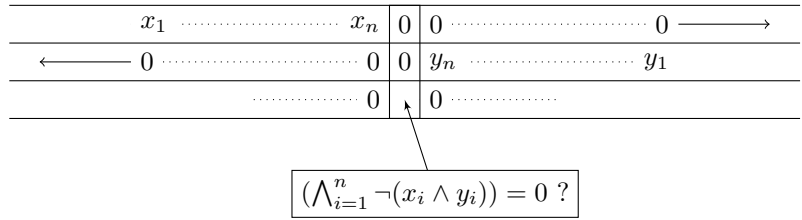
\begin{figure}[ht]
\begin{center}
\begin{tikzpicture}[yscale=0.5,xscale=0.4]
    \draw[dotted](-9,3.5) node[anchor=west,fill=white] {$x_1$}
    --(-0.5,3.5) node[anchor=east, fill=white]{$x_n$};
    \draw[dotted](0.5,3.5) node[anchor=west,fill=white] {$0$}
    --(9,3.5) node[anchor=east, fill=white]{$0$};
    \draw (0,3.5) node {$0$};

    \draw[->](9,3.5)--(12,3.5);
    \draw[->](-9,2.5)--(-12,2.5);
    \draw[dotted](-9,2.5) node[anchor=west,fill=white] {$0$}
    --(-0.5,2.5) node[anchor=east, fill=white]{$0$};
    \draw[dotted](0.5,2.5) node[anchor=west,fill=white] {$y_n$}
    --(9,2.5) node[anchor=east, fill=white]{$y_1$};
    \draw (0,2.5) node {$0$};

    \draw[dotted](-0.5,1.5)node[fill=white,anchor=east]{$0$}--(-5,1.5)
    (0.5,1.5)node[fill=white,anchor=west]{$0$}--(5,1.5);

    \draw (-.5,1) rectangle (0.5,4);

    \draw[-latex'](1,-1)node[rectangle,fill=white,draw=black]
         {$\left(\bigwedge_{i=1}^n\neg(x_i\wedge y_i)\right)=0\ ?$}--(0,1.5);
         \draw(-13,3)--(13,3);
    \foreach\i in {1,2,4}
    \draw(-13,\i)--(13,\i);

\end{tikzpicture}
\end{center}
\caption{An automaton with a hard \cyclp{}{} problem, and an easy
  \invap{}{}.}
\label{fig:hard-cycle}
\end{figure}

This corresponds to the following (partial) rule:
\begin{eqnarray*}
F_3\left(*,\left(\begin{array}{c}*\\ * \\0\end{array}\right),*\right)&=&0\\
F_3\left(*,\left(\begin{array}{c}*\\ * \\1\end{array}\right),*\right)&=&1\\
F_3\left(*,\left(\begin{array}{c}1\\1\\1\end{array}\right),*\right)&=&K
\end{eqnarray*}

We consider a cyclic configuration containing an input for Alice on
the first layer, and an input for Bob on the second layer, (as in
Figure \ref{fig:hard-cycle}), and a third layer everywhere empty,
except for a central ``test'' state, actually performing the tests.
While the test value is $1$, the tests go on. There are three cases:

\begin{itemize}
\item If both Alice and Bob receive the empty set, the configuration is
  $1$-periodic, but Alice and Bob can detect this case with a single bit
  of communication.
\item Else, since the tape is cyclic, if
  $\bigwedge_{i=1}^n\neg(x_i\wedge y_i) = 1$,
  then the test goes on forever,
  producing a (temporal) cycle of length $\Omega(n)$, because in this
  case, at least one $x_i$ or one $y_i$ is $1$, and it is separated from
  the next $1$ (possibly itself !) by at least the $2n+1$ zeros depicted
  on figure \ref{fig:hard-cycle}.

\item Otherwise, the
  test becomes $0$ at some step and a spreading state is generated,
  which erases all the layers in both directions and produce a
  (temporal) cycle of length $1$.
\end{itemize}
Thus, except in the case where both sets are empty,
this \emph{is} an ``implementation'' of the \DISJ problem,
shown in $\Omega(n)$ for several variants of communication complexity in 
\cite{kushilevitz97}. This proves that this automaton can embed an
 $\Omega(n)$ communication problem in some of its configurations,
which is enough to prove that its $\cyclp{}{}$ problem is hard.
  \end{enumerate}
\end{proof}

\begin{remark}
\label{rem:exists}
  We prove in section \ref{hard-cycle} that the last construction of
  proposition \ref{exists} has an \invap{}{} problem in $O(1)$.
\end{remark}

\subsection{Necessary conditions for universality}
\label{sec:ncu}

The following corollary is the main tool provided by this paper to
prove negative results about (intrinsic) universality.

\begin{corollary}
  \label{coro:uni}
  Let $F$ be an intrinsically universal CA. Then it holds that:
  \begin{enumerate}
  \item\label{titi} there exists $u$ s.t. ${\cc{\invap{F}{u}}\in\Omega(n)}$,
  \item\label{toto} ${\cc{\predp{F}}\in\Omega(n)}$,
  \item there exists $k$ s.t. ${\cc{\cyclp{F}{k}}\in\Omega(n)}$.
  \end{enumerate}
  Moreover, if $F$ is only reversible-universal, then \ref{toto} and
  \ref{titi} still holds.
\end{corollary}
\begin{proof}
  It follows from Propositions~\ref{prop:simpred}, \ref{prop:siminva}
  and \ref{prop:simcycl} on the one hand, and Proposition~\ref{exists} on
  the other hand.
\end{proof}

A first application of this corollary to the complexity upper-bounds
presented in Section~\ref{sec:probs} yields the following necessary
conditions for universality. The first proofs of these results appears in
\cite{phd-theyssier}. However, our approach allows us to formulate
much simpler and more elegant proofs.

\begin{corollary}
  Let $F$ be an intrinsically universal CA, then $F$ \emph{cannot be}:
  \begin{itemize}
  \item neither expansive
  \item nor linear
  \item nor reversible.
  \end{itemize}
  Moreover, a reversible universal CA \emph{can not be} expansive or
  linear.
\end{corollary}

\subsection{Uncomparability of the three conditions}
Here we show the ``orthogonality'' of our three problems:
For any pair of problems $({\cal P}_0,{\cal P}_1)$,
we exhibit two CA, ${\cal A}$ and ${\cal B}$, such that:
\begin{itemize}
\item $\cc{{\cal P}_0^{\cal A}}\in o(\cc{{\cal P}_1^{\cal A}})$,
in which case we say that ${\cal A}$ is ``hard'' for ${\cal P}_1$ and
``easy'' for ${\cal P}_0$.
\item $\cc{{\cal P}_1^{\cal B}}\in o(\cc{{\cal P}_0^{\cal B}})$,
in which case we say that ${\cal B}$ is ``hard'' for ${\cal P}_0$ and
``easy'' for ${\cal P}_1$.
\end{itemize}

This shows that our three necessary conditions for intrinsic
universality are \emph{really} necessary:  no condition 
is stronger than any other.

\subsubsection{A CA easy for \predp{} and hard for \invap{}{}}
\label{sec:pred_easy_inv_hard}

The idea is to embed an equality test (more precisely, a \emph{palindrom} test)
launching signals invading the
whole configuration, while keeping the prediction problem easy; see \cite{kushilevitz97}
or proposition~\ref{prop:commlb}
to see why this problem requires $\Omega(n)$ communicated bits. The idea
is to use two components that both stay easy for \predp{}: one with tests
that do not alter the component, and one with signals, moving quickly out
of the way:
\begin{enumerate}

\item The first layer performs tests for equality, as described below, and
  initialy contains a word over the alphabet
  $\Gamma_1=\{\overrightarrow{0},\overrightarrow{1},\overleftarrow{0},
  \overleftarrow{1},\top,\emptyset_1,K_1\}$. On figure
  \ref{fig:easypred-hardinv}, this layer is drawn with full lines.

  The dynamic of the first layer is simple : $\overrightarrow{a}$ states shift
right, and $\overleftarrow{a}$ states shift left. $\top$
states do not move, and $\emptyset_1$ are spreading.
  
\item A layer with an automaton invading the configuration from a
  seed. We need five states on this layer:
  $\Gamma_2=\{s,\emptyset_2,\rightarrow,\leftarrow,K_2\}$. We describe the rule
  below. On figure \ref{fig:easypred-hardinv}, this layer is drawn dashed.

  The rule here is even simpler: $\emptyset_2$ states do not move,
  $\rightarrow$ states shift right, $\leftarrow$
  states shift left. State $s$ represents a signal ``seed'', meaning that if it
  appears once, it disappears on the next step, and changes 
  into a $\rightarrow$ signal on its right, and a $\leftarrow$ signal on its left.

\end{enumerate}

We add a few rules that allow to verify the well-formedness of configurations.
This allows us to ensure that there can be only one $\top$ state on the first
layer, and that signals on the second layer never cross. States $K_1$ and $K_2$
are used for this purpose: if one of them appears somewhere, they \emph{both}
spread on both layers, thus erasing the whole configuration: the \predp{} problem
becomes trivial.

\begin{itemize}
\item If a $\overleftarrow{a}$ state is found immediately next
to an $\overrightarrow{a}$
state, then $K_1$ and $K_2$ are both raised.

\item If a $\rightarrow$ signal is found in the same cell as an $\overrightarrow{a}$,
or a $\leftarrow$ in the same cell as an $\overleftarrow{a}$, then $K_1$ and
$K_2$ are raised. This ensures that signals on the second layer never cross.
\end{itemize}

\begin{figure}[ht]
\begin{center}
\begin{tikzpicture}[scale=0.3]
\def\signals{black!50!white}
\draw[dotted]
(-9.5,0.5) node {\tiny $\overrightarrow{0}$}
(-8.5,0.5) node {\tiny $\overrightarrow{1}$}
(-8,0.3)--(-1,0.3)
(-0.5,0.5) node {\tiny $\overrightarrow{0}$}

(10.5,0.5) node {\tiny $\overleftarrow{0}$}
(9.5,0.5) node {\tiny $\overleftarrow{1}$}
(9,0.3)--(2,0.3)
(1.5,0.5) node {\tiny $\overleftarrow{0}$}

(0.5,0.5) node{\tiny$\top$} 
(0.5,1)--(0.5,5)
(0.5,5.5) node{\tiny$\top$};

\draw(0.3,6) node[anchor=south]{\tiny$\top$};
\draw[dashed, \signals](0.7,6) node[anchor=south]{\tiny$s$};

\draw[dotted]
(0.5,7.5) node{\tiny$\top$}
(0.5,8)--(0.5,10)
(0.5,10.5) node {\tiny$\top$};
\draw(-6,0.5) node[fill=white] {\tiny$\overrightarrow{0}$};
\draw[-latex'](-5.5,1)--(0,6.5);
\draw(7,0.5) node[fill=white]{\tiny$\overleftarrow{1}$};
\draw[-latex'](6.5,1)--(1,6.5);
\draw[-latex',dashed, \signals](1,6.5)--(4,9.5) node[fill=white,shape=circle]{}
(4,9.5) node {\tiny$\rightarrow$}
--(6,11.5);
\draw[-latex',dashed, \signals](0,6.5)--(-3,9.5) node[fill=white,shape=circle]{}
(-3,9.5) node {\tiny$\leftarrow$}
--(-5,11.5);
\draw[ultra thin]
(0,0)--(0,10) (1,0)--(1,10);
\draw
(-10,0)--(-10,1)--(-9,1)--
(-9,1)--(-9,2)--(-8,2)--
(-8,2)--(-8,3)--(-7,3)--
(-7,3)--(-7,4)--(-6,4)--
(-6,4)--(-6,5)--(-5,5)--
(-5,5)--(-5,6)--(-4,6)--
(-4,6)--(-4,7)--(-3,7)--
(-3,7)--(-3,8)--(-2,8)--
(-2,8)--(-2,9)--(-1,9)--
(-1,9)--(-1,10)--(0,10)--
(0,10)--(0,11)--(1,11)--
(1,11)--(1,10)--(2,10)--
(2,10)--(2,9)--(3,9)--
(3,9)--(3,8)--(4,8)--
(4,8)--(4,7)--(5,7)--
(5,7)--(5,6)--(6,6)--
(6,6)--(6,5)--(7,5)--
(7,5)--(7,4)--(8,4)--
(8,4)--(8,3)--(9,3)--
(9,3)--(9,2)--(10,2)--
(10,2)--(10,1)--(11,1)--
(11,1)--(11,0)--cycle;
\end{tikzpicture}
\end{center}
\caption{A CA easy for \predp{} and hard for \invap{}{}}
\label{fig:easypred-hardinv}
\end{figure}
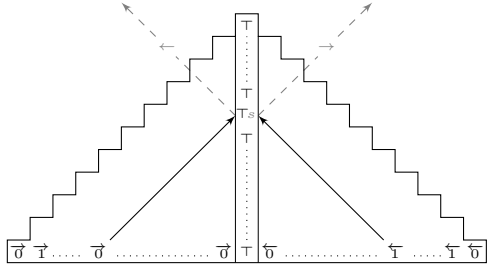

Moreover, we introduce another rule to perform the equality test:
when the test is negative (i.e. a $\top$ state has an
$\overrightarrow{x}$ on its left, a $\overleftarrow{y}$ on its right,
and $x\neq y$), then we place an $s$ state on the second layer :

\begin{center}
\begin{tabular}{ccc}
  $F\left(\begin{array}{c}\emptyset \\ \overrightarrow{a}\end{array},
    \begin{array}{c} \emptyset \\ \top\end{array},
    \begin{array}{c} \emptyset \\ \overleftarrow{a}\end{array}
    \right)$ & $=$ & $\begin{array}{c} \emptyset \\ \top\end{array}$\\
  $F\left(\begin{array}{c}\emptyset \\ \overrightarrow{a}\end{array},
    \begin{array}{c} \emptyset \\ \top\end{array},
    \begin{array}{c} \emptyset \\ \overleftarrow{1-a}\end{array}
    \right)$ & $=$ & $\begin{array}{c} s \\ \top\end{array}$\\
\end{tabular}
\end{center}

\begin{proposition}
  The CA $F$ described above is such that:
  \begin{enumerate}
  \item  ${\predc{F}\in O(1)}$,
  \item there is $u$ such that ${\invac{F}{u}\in\Omega(n)}$.
  \end{enumerate}
  \begin{proof}
    \begin{enumerate}
    \item A protocol for \predp{} needs to predict the content of both layers: if
      the configuration is not well-formed, then a $K_i$ state will appear somewhere
      and this is easy (and it can be checked locally by Alice
      and Bob). Else :
      \begin{itemize}
      \item On the first layer, the result will always be the result of a
        shift if the initial configuration contains only
        $\overrightarrow{x}$ or $\overleftarrow{x}$ states, or if the $\top$
        state is not the central cell of the configuration, and a $\top$
        state else. This requires a constant number of communicated bits.
      \item On the second layer, there are four -- possibly overlapping -- possibilities:
        
        \begin{itemize}
        \item If the leftmost state of Alice's differs from the rightmost
          state of Bob's, and the central cell is a $\top$ state, the result is
          an $s$.
        \item If the $\top$ state is not the central cell, but somewhere
          else in the left part, and the corresponding word is not a palindrom,
          then a $\rightarrow$ is launched (see figure \ref{fig:bla}).
        \item If the initial configuration contained an $s$ or a $\rightarrow$
          in its leftmost cell, a $\rightarrow$ arrives to the top of the triangle.
        \item Else, the result is a $\emptyset_2$.
        \end{itemize}
        All of these can be checked locally and communicated between Alice and Bob within
        a constant number of bits.

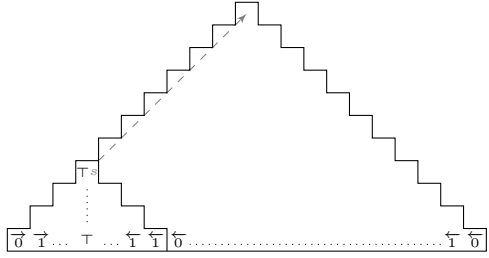
\begin{figure}[ht]
\begin{center}
\begin{tikzpicture}[scale=0.3]
\def\signals{black!50!white}
\draw[dotted]
(-9.5,0.5) node {\tiny $\overrightarrow{0}$}
(-8.5,0.5) node {\tiny $\overrightarrow{1}$}
(-8,0.3)--(-5,0.3)

(10.5,0.5) node {\tiny $\overleftarrow{0}$}
(9.5,0.5) node {\tiny $\overleftarrow{1}$}
(9,0.3)--(-2,0.3)
(-2.5,0.5) node {\tiny $\overleftarrow{0}$};

\draw[dashed, -latex', \signals] (-6.5,3.5)--(0.5,10.5);
\draw[dotted](-6.5,0.5) node[fill=white]{\tiny$\top$} 
(-6.5,1)--(-6.5,3)
(-6.2,3.5) node[\signals, fill=white]{\tiny$s$}
(-6.7,3.5) node{\tiny$\top$};

\draw(-3.5,0.5)node{\tiny$\overleftarrow{1}$};
\draw(-4.5,0.5)node{\tiny$\overleftarrow{1}$};


\draw
(-10,0)--(-10,1)--(-9,1)--
(-9,1)--(-9,2)--(-8,2)--
(-8,2)--(-8,3)--(-7,3)--
(-7,3)--(-7,4)--(-6,4)--
(-6,4)--(-6,5)--(-5,5)--
(-5,5)--(-5,6)--(-4,6)--
(-4,6)--(-4,7)--(-3,7)--
(-3,7)--(-3,8)--(-2,8)--
(-2,8)--(-2,9)--(-1,9)--
(-1,9)--(-1,10)--(0,10)--
(0,10)--(0,11)--(1,11)--
(1,11)--(1,10)--(2,10)--
(2,10)--(2,9)--(3,9)--
(3,9)--(3,8)--(4,8)--
(4,8)--(4,7)--(5,7)--
(5,7)--(5,6)--(6,6)--
(6,6)--(6,5)--(7,5)--
(7,5)--(7,4)--(8,4)--
(8,4)--(8,3)--(9,3)--
(9,3)--(9,2)--(10,2)--
(10,2)--(10,1)--(11,1)--
(11,1)--(11,0)--cycle;

\draw
(-6,4)--(-6,3)--(-5,3)--
(-5,2)--(-4,2)--(-4,1)--
(-3,1)--(-3,0);

\end{tikzpicture}
\end{center}
\caption{A CA easy for \predp{} and hard for \invap{}{}}
\label{fig:bla}
\end{figure}
\end{itemize}

\item Now we need to find a set of hard instances for the \invap{}{}
problem: with a background word $u$, with $\emptyset_i$ on both layers, and
an initial configurations of the form
$(\overrightarrow{0},\overrightarrow{1})^n\top
(\overleftarrow{0},\overleftarrow{1})^n$ on the first layer, and
$\emptyset^*$ on the second, we reduce the equality problem to
\invap{}{}.
\end{enumerate}
\end{proof}
\end{proposition}

\subsubsection{A CA easy for \cyclp{}{} and hard for \invap{}{}}
We can reuse the construction of paragraph
\ref{sec:pred_easy_inv_hard}: we already know that it is hard for \invap{}{}.
What we need to do is to modify the rule so that on the second layer, when a
$\rightarrow$ signal crosses a $\leftarrow$ signal, they both disappear and
the resulting state is a $\emptyset_2$. This ensures that on cyclic configurations,
even if signals are ``raised'' somewhere, they are ``caught'' by the cyclicity.
The rest of the discussion is essentially the same as in paragraph
\ref{sec:pred_easy_inv_hard}, and we can conclude easily that the orbits of configurations
containing at least one $\top$, or of ill-formed configurations,
are always 1-periodic; the \cyclp{}{} problem can be decided with no communication.
In all other cases, the dynamic is nothing more than a shift: the protocol from
\ref{prop:reversible} can be used.

\subsubsection{A CA easy for \predp{}, and hard for \cyclp{}{}}
\label{hard-cycle-easy-pred}

We can use once again (and for the last time) quite the same construction as in
paragraph \ref{sec:pred_easy_inv_hard}. We modify it to launch only one
signal (in only one direction) when an error appears. Thus, as proven
in section \ref{sec:pred_easy_inv_hard}, the \predp{} problem remains
easy. Now we need to prove that the \cyclp{}{} problem is hard, but for this we
can choose the instances on purpose.

If no test fails, the configuration will be 1-periodic: When all the
tests have been done, the configuration is uniformly empty, except for
the $\top$ states, and then nothing more happens. Otherwise, a signal will
be launched. We need to show that the period of the configuration is
then in $\Omega(n)$. But we can notice that a contiguous portion of
$\Omega(n)$ cells can not have any signal (see Figure
\ref{fig:pred_easy_cycl_hard}). Therefore, the period of the
configuration is $\Omega(n)$ if and only if an error occurs.

\begin{figure}[htb]
\begin{center}
\begin{tikzpicture}[scale=0.3]
\draw[dotted] 
(-9.5,0.5) node {\tiny $\overrightarrow{0}$}
(-8.5,0.5) node {\tiny $\overrightarrow{1}$}
(-8,0.3)--(-1,0.3)
(-0.5,0.5) node {\tiny $\overrightarrow{0}$}

(10.5,0.5) node {\tiny $\overleftarrow{0}$}
(9.5,0.5) node {\tiny $\overleftarrow{1}$}
(9,0.3)--(2,0.3)
(1.5,0.5) node {\tiny $\overleftarrow{0}$}

(0.5,0.5) node{\tiny$\top$} 
(0.5,1)--(0.5,5)
(0.5,5.5) node{\tiny$\top$};

\draw(0.3,6) node[anchor=south]{\tiny$\top$};
\draw[blue!50!white](0.7,6) node[anchor=south]{\tiny$s$};

\draw[dotted]
(0.5,7.5) node{\tiny$\top$}
(0.5,8)--(0.5,10)
(0.5,10.5) node {\tiny$\top$}
(0.5,11)--(0.5,16.5);
\draw(-6,0.5) node[fill=white] {\tiny$\overrightarrow{0}$};
\draw[-latex'](-5.5,1)--(0,6.5);
\draw(7,0.5) node[fill=white]{\tiny$\overleftarrow{1}$};
\draw[-latex'](6.5,1)--(1,6.5);
\draw[-latex',blue!50!white](1,6.5)--(4,9.5) node[fill=white,shape=circle]{}
(4,9.5) node {\tiny$\rightarrow$}
--(6,11.5);

\draw[ultra thin]
(0,0)--(0,10) (1,0)--(1,10);
\draw
(-10,0)--(-10,1)--(-9,1)--
(-9,1)--(-9,2)--(-8,2)--
(-8,2)--(-8,3)--(-7,3)--
(-7,3)--(-7,4)--(-6,4)--
(-6,4)--(-6,5)--(-5,5)--
(-5,5)--(-5,6)--(-4,6)--
(-4,6)--(-4,7)--(-3,7)--
(-3,7)--(-3,8)--(-2,8)--
(-2,8)--(-2,9)--(-1,9)--
(-1,9)--(-1,10)--(0,10)--
(0,10)--(0,11)--(1,11)--
(1,11)--(1,10)--(2,10)--
(2,10)--(2,9)--(3,9)--
(3,9)--(3,8)--(4,8)--
(4,8)--(4,7)--(5,7)--
(5,7)--(5,6)--(6,6)--
(6,6)--(6,5)--(7,5)--
(7,5)--(7,4)--(8,4)--
(8,4)--(8,3)--(9,3)--
(9,3)--(9,2)--(10,2)--
(10,2)--(10,1)--(11,1)--
(11,1)--(11,0)--cycle;


\begin{scope}[ultra thin]

\draw(0,0)--(12.5,0);
\draw(6,5.5)--(12.5,5.5);
\draw(0,16.5)--(12.5,16.5);

\draw (12,2.75)node[anchor=west]{\small $n/4$ : No signals here};
\draw (12,11)node[anchor=west]{\small $\leq n/2$ signals};

\draw[latex'-latex'] (12,0)--(12,5.5);
\draw[latex'-latex'](12,5.5)--(12,16.5);

\end{scope}
\end{tikzpicture}

\end{center}
\caption{A CAeasy for \predp{} and hard for \cyclp{}{}.}
\label{fig:pred_easy_cycl_hard}
\end{figure}
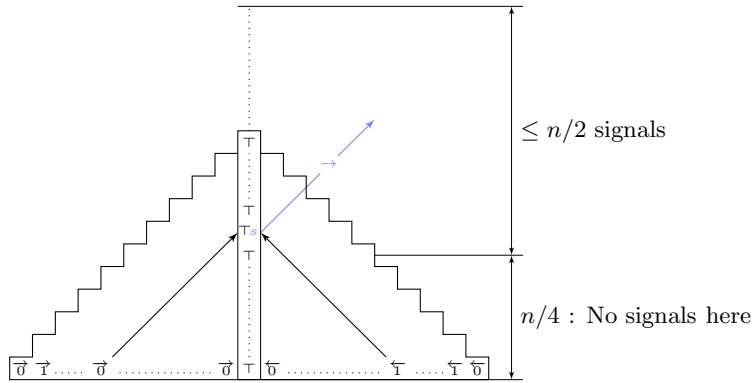

\subsubsection{A CA easy for \invap{}{} and hard for \cyclp{}{}}
\label{hard-cycle}
As promised in remark \ref{rem:exists}, we now prove a protocol for the \invap{}{}
problem of the rule described there:
\begin{proposition} The CA $F$ described in the proof of proposition \ref{exists} is such that:
  \[{\forall u, \invac{F}{u}\in O(1)}\]
\begin{proof}
  Let $u$ be any word over the alphabet for $F$. First,
  if the orbit of $p_u$ contains a spreading state, then $p_u(w)$ quickly becomes
  uniform with the spreading state everywhere, independently from $w$.
  Else, the discussion is a little more subtle.
  Let us note the periodic background $p_u=(p_{u_1},p_{u_2},p_{u_3})$, and let
  $w=(w_1,w_2,w_3)$ the input, split between Alice and Bob.
 
  \begin{enumerate}
  \item If $p_{u_1}(w_1)\neq p_{u_1}$, and $p_{u_2}(w_2)\neq p_{u_2}$,
  and then either a spreading state is generated, or
  the differences on components one and two are shifted in opposite
  directions, thus also invading $p_u$.

  \item \label{hard-cycle-eq-case}
  If $p_{u_1}(w_1)=p_{u_1}$ and $p_{u_2}(w_2)=p_{u_2}$, maybe the third
  component (the actual ``tests'') changes between $p_u$ and $p_u(w)$, but then there is an
  easy way to transmit whole configurations : Alice can simply tell Bob that her part
  is the same as in $p_{u}$, on the first two components.
  If Bob does the same, then both know both ``sets'', and
  they can check without more communication if their respective portions
  of $p_{u_3}(w_3)$ ever generates a spreading state : if so, $p_u(w)$ is invaded,
  else it is not.

  \item Else, without loss of generality, we can assume that $p_{u_1}(w_1)=p_{u_1}$
  and $p_{u_2}(w_2)\neq p_{u_2}$. There are two cases :

  \begin{itemize}
  \item Either $p_{u_3}(w_3)=p_{u_3}$ (the ``tests'' are the same in $p_u$ and $p_u(w)$),
  and then using the trick from (\ref{hard-cycle-eq-case}), Alice and Bob
  can know $p_{u_1}(w_1)$ and $p_{u_3}(w_3)$ completely, within constant communication.

  Then, since they each know a part of set $p_{u_2}(w_2)$, and they both know $p_{u_3}(w_3)$,
  they can check disjointness with $p_{u_1}(w_1)$ separately and tell if a spreading
  state ever appears, which is the only way $p_u(w)$ can be invaded in this case.

  \item If $p_{u_3}(w_3)\neq p_{u_3}$, then either a spreading state is generated,
  or $p_{u_3}(w_3)$ stays fixed, and $p_{u_2}(w_2)$ shifts to infinity:
  in both cases, $p_u(w)$ is invaded.

  \end{itemize}
  \end{enumerate}
\end{proof}
\end{proposition}

\subsubsection{A CA easy for \invap{}{} and hard for \predp{}}
\label{easy-inv-hard-pred}
Elementary rule 218 is a natural example exhibiting this property.
Unfortunately, the proof is quite technical and requires an in-depth
study of rule 218, which we chose to delay until section \ref{sec:218},
for conciseness of this --already long-- section, and consistency of
section \ref{sec:concrete}.

\subsubsection{An CA easy for \cyclp{}{} and hard for \predp{}}

We describe the natural example of Rule 33 in Section \ref{sec:33}, 
which has a protocol in constant time for \cyclp{}{}, and for which any
deterministic protocol for \predp{} is in $\Omega(\log n)$.

\section{Intrinsic universality: Ruling out complex CA} 
\label{sec:iu}
Here we show that for
two of our canonical problems -- namely, \predp{} and \invap{}{} -- we were
able to find a CA of maximal algorithmic complexity (\emph{complete}), 
and yet very simple with respect to
our framework. 

More precisely, we are going to show that, for problems \predp{} and
\invap{}{}, there exists a CA $F$ for which the communication
complexity of the problem is low while its classical computational
complexity is the highest one can expect.

Therefore, we are ruling out such non-trivial CA from being
intrinsically universal.

\subsection{Prediction}
T. Neary and D. Woods proved ``the \textsc{P}-completeness of Rule
110'' \cite{woodsneary06}.  In our language, they proved that the
problem \predp{F_{110}} is \textsc{P}-complete. A very
natural question arises: What do classical algorithmic properties of
CA, such as \textsc{P}-completeness, imply on their communication
complexity counterpart?  

As we show in this section, such a strong computational property
is not enough to guarantee maximal communication complexity. However, we do
not know of an automaton that would have, for instance,
polylogarithmic communication complexity, and still a \textsc{P}-complete
prediction problem, nor do we have a nonexistence proof. We leave this
as an open problem.

\begin{proposition}
\label{prop:pcomplete}
For any ${k\geq 1}$, there exists a CA $F$ such that
\[\cc{\predp{F}}\in O( n^{1/k})\]
and \predp{F} is \textsc{P}-complete.

\begin{proof}

Let $\mathcal{M}$ a Turing machine. We construct a CA $F$ simulating
$\mathcal{M}$ slowly but still in polynomial time: it takes $n^k$ steps of
$F$ to simulates $n$ steps of $\mathcal{M}$. Hence, by a suitable choice of
$\mathcal{M}$, the problem of predicting $F$ is P-complete.

First it is easy to construct a CA simulating $\mathcal{M}$ in real time. We
encode each symbol of the tape alphabet of the Turing machine by a CA state, and
add a ``layer'' for the head, with '$\rightarrow$' symbols on its left and
'$\leftarrow$' symbols on its right. We guarantee this way that there can be
only one head: if a '$\rightarrow$' state is adjacent to a '$\leftarrow$' state
without a head between them, we propagate a spreading ``error'' state destroying
everything.

We then add a new layer to slow down the simulation: it consists in a single
particle (we use the same trick to ensure that there is only one particle)
moving left and right inside a marked region of the configuration. More
precisely, it goes right until it reaches the end of the marked region, then it
adds a marked cell at the end and starts to move left to reach the other end,
doing the same thing forever. Clearly, for any cell in a finite marked region,
seeing $n$ traversals of the particle takes $\Omega(n^2)$ steps. Then, the idea
is to authorize head moves, in the previous construction, only at particle
traversals. This way, $n$ steps of $\mathcal{M}$ require $n^2$ time steps of the
automaton. By adding another particle layer, one can also slow down the above
particle with the same principle and it is not difficult to finally construct a
CA $F$ such that $n$ steps of $\mathcal{M}$ require $n^k$ time steps of
$F$. We have represented in Figure \ref{fig:pcomplete} the behavior of the
particle, with the dashed arrow representing a Turing transition.

Now if the initial configuration does not respect the rules described above,
then a spreading error state is generated and Alice and Bob can notice it within
constant communication. In all other cases, it is enough for Alice or Bob to
know the value of all the $2\cdot n^{1/k}$ states around the initial position of
the head, because the computation of the Turing machine simply does not depend
on the rest of the initial configuration. So for these cases, at most $n^{1/k}$
bits need to be communicated for Alice or Bob to compute the answer. Note that
if the bounds for the particle are absent from the initial configuration, then
no transition can happen, thus Alice and Bob know the result in constant time.

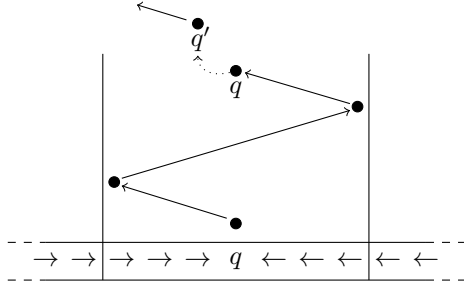
\begin{figure}[ht!]
\begin{center}
\begin{tikzpicture}

\draw(0,0)--(5,0) (0,0.5)--(5,0.5);
\draw[dashed](-0.5,0)--(0,0) (5,0)--(5.5,0) (-0.5,0.5)--(0,0.5) (5,0.5)--(5.5,0.5);
\foreach\x in {0,0.5,1,1.5,2} {\draw (\x,0.25)node{$\rightarrow$};}
\draw(2.5,0.25)node{$q$};
\foreach\x in {3,3.5,4,4.5,5} {\draw (\x,0.25)node{$\leftarrow$};}
\draw(4.25,0)--(4.25,3) (0.75,0)--(0.75,3); 

\draw[fill=black](2.5,0.75)circle(0.07cm);
\draw[->](2.38,0.82)--(1,1.25);
\draw[fill=black](0.9,1.3)circle(0.07cm);
\draw[->](1,1.35)--(4,2.25);
\draw[fill=black](4.1,2.30)circle(0.07cm);
\draw[->](4,2.34)--(2.62,2.75);
\draw[fill=black](2.5,2.775)circle(0.07cm);
\draw(2.5,2.525)node{$q$};
\draw[->,dotted](2.5,2.775) .. controls (2.2,2.675) and (2,2.775)..(2,2.975);

\draw[fill=black](2,3.4)circle(0.07cm);
\draw(2.04,3.2) node {$q'$};
\draw[->](1.85,3.45)--(1.2,3.65);
\end{tikzpicture}
\end{center}
\caption{A CA for which \predp{} is P-complete.}
\label{fig:pcomplete}
\end{figure}
\end{proof}

\end{proposition}

\begin{remark}
  A result by Hromkovic (see \cite{hromkovic97}) states that a Turing
  machine with a single head working in time $t(n)$ can only recognize
  a language of communication complexity less than $O(\sqrt{t(n)})$.
  Said differently, a CA simulating a Turing machine cannot produce
  instances of communication complexity more than $O(\sqrt{n})$ for
  the prediction problem on configurations with a single head
  (whatever the machine does).
\end{remark}

\subsection{Invasion}
\label{subsubsec:undecidable}
This problem is even more complex than \predp{}: It is in fact undecidable.
However, since there is no limitation on the ``classical'' computational power of
Alice and Bob, it can still be decided within very little communication.
\begin{proposition}\ 
  \par
  \begin{enumerate}
  \item For any CA $F$ and any word $u$, we
    have: ${\invap{F}{u}\in\Pi_1^0}$.
  \item Their exist $F$ and $u$ such that $\invap{F}{u}$ is
    $\Pi_1^0$-complete, and yet $\invac{F}{u}\in O(\log n)$
  \end{enumerate}
  \begin{proof}\hfill
  \begin{enumerate}
  \item Let $F$ and $u$ be fixed and consider the problem
    $\invap{F}{u}$. Given an input ${x_1,\ldots, x_n}$,
    we use the notations $\lediff{t}$ and $\ridiff{t}$ for the
    leftmost and righmost differences at time $t$ between the orbit of
    $p_u$ and the orbit of $p_u(x_1\cdots x_n)$ as in
    Definition~\ref{def:invasion}. 
    \begin{claim}
      There exists a recursive function $\beta$ such that for any $n$,
      any input $x_1,\ldots,x_n$ and any $\Delta\geq 0$ we have:
      \[\exists t,\, \ridiff{t}-\lediff{t}\geq\Delta \, \iff\ \exists
      t\leq\beta(\Delta),\, \ridiff{t}-\lediff{t}\geq\Delta.\]
    \end{claim}
    The proof follows from the above claim because the invasion problem can be
    expressed as the following $\Pi_1^0$ predicate:
    \[\forall\Delta\geq 0,\, \underbrace{\exists t\leq\beta(\Delta),\,
      \ridiff{t}-\lediff{t}\geq\Delta}_{\text{recursive predicate}}\]
    \begin{proof}[Proof of the claim]
      First, the orbit of $p_u$ is ultimately periodic: There are
      $t_0$ and $p$ such that for any $t\geq t_0$ we have
      ${F^t(p_u)=F^{t+p}(p_u)}$. Given an input ${x_1,\ldots, x_n}$ of
      the problem, denote by $w(t)$ the word of length
      ${\ridiff{t}-\lediff{t}}$ starting at position $\lediff{t}$ in
      configuration ${F^t\bigl(p_u(x_1,\ldots,x_n)\bigr)}$. The key
      point is that for any $t\geq t_0$, the triple
      \[\chi(t+1)=\bigl(w(t+1),\lediff{t+1}\bmod |u|, t+1\bmod p\bigr)\]
      is uniquely determined by the triple
      \[\chi(t)=\bigl(w(t),\lediff{t}\bmod |u|, t\bmod p\bigr)\]
      (because the word $w(t)$ ``evolves'' in a periodic context and
      knowing the offset of the position of $w(t)$ in that context is
      enough to know $w(t+1)$). Therefore, if the words $w(t)$ are
      bounded by $\Delta$ for a sufficiently long time (exponential in
      $\Delta$), then the triple $\chi(t)$ will take a value already
      taken before and the sequence $\bigl(\chi(t)\bigr)_t$ will be
      ultimately periodic, showing that ${|w(t)|}$ is bounded and that
      there is no invasion.  Adding $t_0$ to this exponential function
      is a convenient choice for $\beta$.
    \end{proof}
  \item We build a CA $F$ that simulates a 2-counter machine
  \cite{minsky}. More precisely, standard states have two layers: a
  data layer over states ${A,M,B,0}$, used to store the value of the 2
  unary counters, and a control layer made of a Turing head storing a
  state from $Q$, with the extra $\rightarrow$ and $\leftarrow$
  symbols ensuring the uniqueness of the head. Finally, $F$ possesses
  a blank state $\emptyset$ and a spreading state $K$ to deal with
  encoding problems. The state set is therefore
  \[\bigl(Q\cup\{K,\emptyset,\rightarrow,\leftarrow\}\bigr)\times\{A,B,0,M\}.\]

  A valid configuration is a configuration everywhere equal to
  $\emptyset$ except on finite coding segments which have the folloing
  form (see figure~\ref{fig:well-formed-conf}):
  \begin{itemize}
  \item the data layer must be of the form: ${0^\ast
      A^+MB^+0^\ast}$;
  \item the control layer must be of the form:
    ${\rightarrow^+q\leftarrow^+}$ with $q\in Q$.
  \end{itemize}

  \begin{figure}[ht]
    \begin{center}
      \begin{tikzpicture}
        \draw[dashed](-3.25,0)--(-0.5,0) (5,0)--(5.75,0)
        (-3.25,1)--(-0.5,1) (5,1)--(5.75,1); \draw(-0.5,0)--(5,0)
        (-0.5,1)--(5,1); \draw[lightgray](-0.5,0.5)--(5,0.5);
        \draw[dashed,lightgray](-3.25,0.5)--(-0.5,0.5)
        (5,0.5)--(5.75,0.5); \draw(-2,0.25) node{Data layer} (-2,0.75)
        node{Control layer}; \draw(0,0.25) node {$\emptyset$} (0,0.75)
        node {$\emptyset$}; \draw(0.25,0)--(0.25,1); \draw(0.5,0.25)
        node {$\emptyset$} (0.5,0.75) node {$\emptyset$};
        \draw(1,0.25) node {$A$} (1,0.75) node {$\rightarrow$};
        \draw(1.5,0.25) node {$A$} (1.5,0.75) node {$\rightarrow$};
        \draw(2,0.25) node {$A$} (2,0.75) node {$\rightarrow$};
        \draw(2.5,0.25) node {$M$} (2.5,0.75) node {$q$};
        \draw(3,0.25) node {$B$} (3,0.75) node {$\leftarrow$};
        \draw(3.5,0.25) node {$0$} (3.5,0.75) node {$\leftarrow$};
        \draw(4,0.25) node {$\emptyset$} (4,0.75) node {$\emptyset$};
        \draw(4.5,0.25) node {$\emptyset$} (4.5,0.75) node
        {$\emptyset$}; \draw(4.75,0)--(4.75,1); \draw(5,0.25) node
        {$\emptyset$} (5,0.75) node {$\emptyset$};
        
      \end{tikzpicture}
    \end{center}
    \caption{A well-formed piece of configuration. The counter $A$
      contains value $3$ and the counter $B$ contains value $1$ in this
      example.}
    \label{fig:well-formed-conf}
  \end{figure}
  
  The number of $A$s and $B$s represent the current value of the 2
  counters. The behaviour of $F$ is the following:
  \begin{itemize}
  \item If the configuration is not valid (which can be detected
    locally), then the state $K$ is generated and spreads;
  \item If the configuration is valid, then on each coding segment,
    the (necessarily unique) head goes repeatedly from one end
    of the segment to the other end, and extends the segment at each pass
    by adding a $\rightarrow$ on the left (resp. $\leftarrow$ on the
    right) and a $0$ on the data layer. If the extension step is
    blocked by another segment, then the state $K$ is generated and
    spreads;
  \item Moreover, at each pass on the segment, the head executes
    one of the basic 2-counter machine's instructions:
    \begin{itemize}
    \item testing if a counter is empty can be done by checking if
      there is a $0$ on the right (resp. the left) of the unique $M$;
    \item decrementing can be done be replacing the leftmost $A$
      (resp. rightmost $B$) by a $0$;
    \item incrementing can be done by replacing a $0$ by $A$ on the
      left of the leftmost $A$ (resp. by $B$ on the right of the
      rightmost $B$); there must be a $0$, because the segment is
      extended at each passage by both sides;
    \item finally, the head can simply stop.
    \end{itemize}
  \end{itemize}

  If any order given to the head leads to an incoherence (decrement an
  empty counter, write a $B$ when on the '$A$' part of the segment,
  etc), the state $K$ is genereated and spreads.

  With this definition, and if $u=\emptyset$, the halting problem for
  the 2-counter machine encoded in $F$ (input: value of counters;
  output: does it halt started from these values ?) clearly reduces to
  $\invap{F}{u}$ (halt$\iff$ no invasion). Therefore, by a suitable
  choice of the 2-counter machine used to construct $F$, we have that
  $\invap{F}{u}$ is $\Pi_1^0$-complete.

  To conclude the proof, we show that ${\cc{\invap{F}{u}}\in
    O(\log(n))}$. Given an input $w$ split between Alice and Bob,
  the following protocol determines whether ${\invap{F}{u}(w)=1}$:
  \begin{itemize}
  \item first Alice and Bob check whether the input configuration is
    valid; if not, the answer is 'invasion'; this can be done with
    $O(1)$ bits of communication since validity is a local property;
  \item the configuration being valid, Alice and Bob communicate so
    that for any pair of consecutive valid segments $s_1$ and $s_2$,
    either Alice or Bob knows the state of both $s_1$ and $s_2$ and
    the distance between them; to achieve this, even if a segment is
    split between Alice's part and Bob's part, it is sufficient
    that they communicate $O(\log(n))$ bits; indeed, a segment is
    completely defined by:
    \begin{itemize}
    \item the value and position of the head,
    \item number of $0$ states on the right and the same on the left,
    \item number of $A$s and number of $B$s.
    \end{itemize}
  \item since for each pair of valid segment, Alice or Bob as enough
    information to detect a possible future collision, they can
    determine together with $O(1)$ bits of communication whether there
    is invasion or not; indeed, invasion is equivalent to: either
    their is a collision somewhere, or their is a single segment
    holding a non-halting computation.
  \end{itemize}

\end{enumerate}
\end{proof}

\end{proposition}

\subsection{Cycle-length}
For this problem, we could find a CA of maximal
algorithmic complexity, as shown by the following proposition.
However, we have to leave as an open problem the existence of a CA
$F$ for which both
$\cyclp{F}{k}$ is \textsc{pspace}-complete for some $k\in\setN$,
and $\cyclc{F}{k}\in o(n)$.

\begin{proposition}
  \par
  \label{prop:pspacecomplete}
  \begin{enumerate}
  \item For any CA $F$ and any $k\geq 1$, $\cyclp{F}{k}\in\textsc{pspace}$.
  \item Their exist $F$ and $k$ such that $\cyclp{F}{k}$ is
    \textsc{pspace}-complete.
  \end{enumerate}
\end{proposition}

\begin{proof}\hfill
  \begin{enumerate}
  \item Let $F$ and ${k\geq 1}$ be fixed. The length of the cycle
    reached by iterating $F$ on a periodic initial configuration $c$
    can be determined in polynomial space with the algorithm described
    below.  Let $n$ be the period of $c$. Starting from $c$, the cycle
    is reached in less than $\alpha^n$ steps where $\alpha$ is the
    cardinal of the state set.
    \begin{enumerate}
    \item compute $c_0=F^{\alpha^n}(c)$ (memory usage: $O(n)$);
    \item memorize $c_0$ and compute the first $t$ such that
      $F^t(c_0)=c_0$ (memory usage: $O(n)$ because such a $t$ is less
      than $\alpha^n$).
    \end{enumerate}
  \item 
  To show this, we embed a Turing machine $\mathcal{M}$,
  deciding a \textsc{pspace}-complete language,
  in a cyclic configuration for a cellular automaton.
  $\mathcal{M}$ works in polynomial space, meaning that there is a
  polynomial $P\in \setN[X]$ such that for any $x\in\Gamma^*$, it will
  never use more than $P(|x|)$ tape cells.
\begin{figure}[htb]
\begin{center}
\begin{tikzpicture}[scale=0.7]
\pgfmathdeclarerandomlist{color}{{yellow}{green}{blue}{red}}

\foreach\alpha in{0,...,\time}{
  \pgfmathrandomitem{\c}{color}
}
\def\start{30}
\def\step{40}
\def\ending{100}

\foreach\alpha in{\start,\step,...,\ending}{
  \pgfmathrandomitem{\c}{color}
  \draw[draw=none, fill=\c!50!white](\alpha:2.5)--(\alpha:3) arc (\alpha:10+\alpha:3)
    -- (10+\alpha:2.5) arc (10+\alpha:\alpha:2.5);
}
\draw[->](5+\ending:3.5) node[anchor=east] {Turing head (state: $Q$)} --(5+\ending:3) ;
\foreach\alpha in{0,10,...,350}{
  \draw(\alpha:2.5)--(\alpha:3);
}
\draw(0,0) circle (3) (0,0) circle (2.5);
\draw[<-](\ending+10:2.2) arc (\ending+10:\ending+200:2.2);
\draw[->](\ending+200:2.2) arc (\ending-160:\ending+10:2.2);
\draw(\ending+15:2.2) node[anchor=north west] {$P(n)$};
\end{tikzpicture}
\end{center}
\caption{The output of the transducer used in Proposition
  \ref{prop:pspacecomplete}.
\label{fig:pspace}
}

\end{figure}
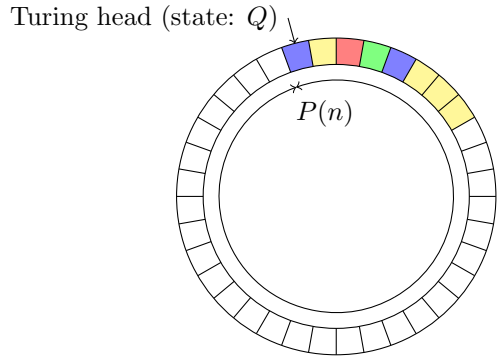

  We can encode a Turing machine easily into a simple cellular
  automaton F{}: the states code for the Turing tape cells, and there
  is a special ``head'' state carrying the state of the machine.  It
  can be easily shown that we can encode the transitions of a Turing
  machine into a local cellular automaton rule, ensuring that if there
  is only one head at the beginning, then it will be so during all the
  computation.

  Moreover, the accepting state is \emph{spreading}, meaning that if
  it appears somewhere, it spreads over all the configuration in both
  directions. The rejecting state launches a particle erasing the
  configuration (i.e., writing blank states everywhere), but shifting
  clockwise. In this way, an accepting computation will result in period
  1, whereas rejecting computations will yield periods of the size of
  the configuration.

  A polynomial-time transducer can easily encode an input $x$ for
  $\mathcal{M}$ into a (cyclic) configuration of F, like shown in
  figure \ref{fig:pspace}. It first directly translates $x$ into states
  of F, then computes $P(|x|)$ and outputs $P(x)$ blank states.
\end{enumerate}
\end{proof}

\section{Intrinsic universality: Ruling out concrete elementary CA}
\label{sec:concrete}
\subsection{CA Rule 218}
\label{sec:218}

The local function 
$f_{218}: \{0,1\}^3 \to \{0,1\}$
of CA rule 218 is defined in
Figure~\ref{fig:rule-218}.

\begin{figure}[htb]
\begin{center}
\subfigure[$F_{218}.$]
{
{\carule 0 1 0 1 1 0 1 1}
\label{fig:rule-218}
}
\subfigure[Example of a space-time diagram for CA Rule 218.]
{
\docompile{\input{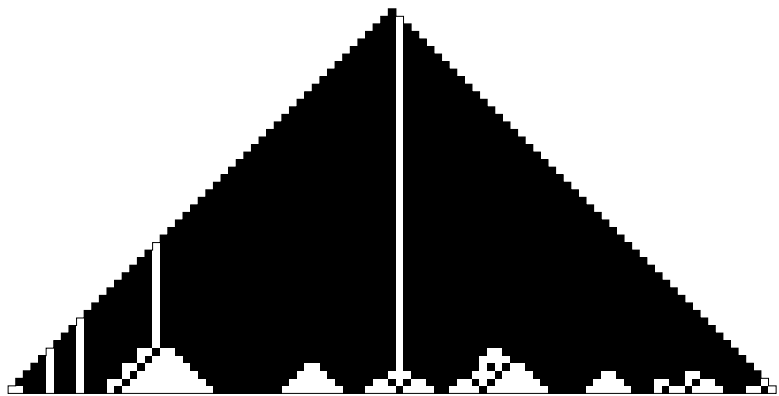}}
\label{spacetime-218}
}
\end{center}
\caption{CA rule 218.}
\end{figure}

From the result of \cite{rap08} we already knew that
 $\cc{\predp{F_{218}}} \in O(\log(n))$.  
It follows from Corollary~\ref{coro:uni} that Rule 218
is not intrinsically universal.  Nevertheless, the proof of
\cite{rap08} was very long and complicated.  As we are going to see
now, the invasion approach gives a short and elegant proof of the same
result.

\begin{definition}
  A word is \emph{additive} if 1s are isolated and separated by an odd
  number of 0s. By extension, an infinite configuration is
  \emph{additive} if it contains only additive words.
\end{definition}

\begin{lemma}
  \label{218:additiv}
  Additivity is preserved by iterations. Moreover, if $abc$ is
  additive then:
  \[f_{218}(a,b,c)\not=f_{218}(1-a,b,c)\text{ and
  }f_{218}(a,b,c)\not=f_{218}(a,b,1-c).\]
\end{lemma}
\begin{proof}
  First additivity is preserved by iterations because $010^n10$
  becomes $010^{n-2}10$ for $n\geq 3$ and $01010$ becomes $000$.

  To conclude the lemma, it is sufficient to check that, for any
  $a$,$b$,$c$ such that $11$ is not a factor of $abc$ then:
  \[f_{218}(a,b,c)\not=f_{218}(1-a,b,c)\text{ and
  }f_{218}(a,b,c)\not=f_{218}(a,b,1-c).\]
\end{proof}

\begin{lemma}
  \label{218:makewall}
  Let $c$ be any non-additive configuration. Then, after a finite
  time, the word $11$ appears in the evolution and this word is a
  wall.
\end{lemma}
\begin{proof}
  First $11$ is a wall because:
  \[f_{218}(\ast,1,1)=f_{218}(1,1,\ast)=1.\] To conclude it is sufficient to
  check that the image of ${10^n1}$ with ${n\geq 2}$ is ${10^{n-2}1}$.
\end{proof}

\begin{proposition}
  For all $u$, we have $\cc{\invap{F_{218}}{u}}\leq 1$.
\end{proposition}
\begin{proof}
  First, if the configuration $\perio{u}$ is non-additive then, by
  Lemma~\ref{218:makewall}, at some time $t$ a wall appears
  periodically in ${F_{218}^t(\perio{u})}$. Hence, for any
  ${x_1,\ldots, x_n}$, the differences between
  $\perio{u}(x_1,\ldots,x_n)$ and $\perio{u}$ are bounded to a fixed
  finite region. Said differently, there is never
  propagation for such an $u$.

  Now consider the case where $\perio{u}$ is additive. By
  Lemma~\ref{218:additiv}, we have for any ${x_1,\ldots, x_n}$:
  \begin{itemize}
  \item either ${\perio{u}=\perio{u}(x_1,\ldots,x_n)}$,
  \item or for any $t\geq 0$:
    \begin{align*}
      \lediff{t} &= \lediff{0}-t\\
      \ridiff{t} &= \ridiff{0} + t
    \end{align*}
  \end{itemize}
  
Therefore, the problem consists in deciding whether 
${\perio{u}}$ and ${\perio{u}(x_1,\ldots,x_n)}$ are equal, 
which can be done
 with $1$ bit of communication.
\end{proof}

\begin{corollary}
CA Rule 218 is not intrinsically universal.
\end{corollary}

As promised in section \ref{easy-inv-hard-pred},
it remains to show that the deterministic (possibly with several rounds)
communication complexity of the \predp{} problem for rule 218 is ``hard'',
according to our conventions:

\begin{proposition}
\[\predp{F_{218}}\in\Omega(\log n)\]
\begin{proof}
To show this, we construct a fooling set $S_n$ (see
Definition \ref{def:fool} or \cite{kushilevitz97}):
\[S_n=\{(1^{n-k}0^k,0^{k+1}1^{n-k},0\leq k\leq n\}\] 

We show that $S_n$ is a fooling set for Rule $218$: In fact, on all
configurations of the form $1^{n-k}0^{2k+1}1^{n-k}$, the result of
\predp{F_{218}} is always $0$. On configurations of the form
$1^{n-i}0^{i+j+1}1^{n-j}$ where $i\neq j$, it is always $1$. This can
be easily shown from the collection of lemmas of \cite{rap08}, and we
illustrate it on Figure \ref{fig:218_fool}. Thus, since $|S_n|=n+1$,
we deduce that a deterministic protocol solving \predp{F^n_{218}} can
not take less than $\log(n+1)$ steps:
\[\cc{\predp{F_{218}^n}}\in\Omega(\log(n))\]
\begin{figure}[htb]
\begin{center}
\begin{tikzpicture}[scale=0.3]
\draw[draw=none,fill=black](-6,0)--(-5,0)--(-5,1)--(-6,1)--cycle;
\draw[draw=none,fill=black](-5,0)--(-4,0)--(-4,1)--(-5,1)--cycle;
\draw[draw=none,fill=black](-4,0)--(-3,0)--(-3,1)--(-4,1)--cycle;
\draw[draw=none,fill=white](-3,0)--(-2,0)--(-2,1)--(-3,1)--cycle;
\draw[draw=none,fill=white](-2,0)--(-1,0)--(-1,1)--(-2,1)--cycle;
\draw[draw=none,fill=white](-1,0)--(0,0)--(0,1)--(-1,1)--cycle;
\draw[draw=none,fill=white](0,0)--(1,0)--(1,1)--(0,1)--cycle;
\draw[draw=none,fill=white](1,0)--(2,0)--(2,1)--(1,1)--cycle;
\draw[draw=none,fill=white](2,0)--(3,0)--(3,1)--(2,1)--cycle;
\draw[draw=none,fill=white](3,0)--(4,0)--(4,1)--(3,1)--cycle;
\draw[draw=none,fill=white](4,0)--(5,0)--(5,1)--(4,1)--cycle;
\draw[draw=none,fill=white](5,0)--(6,0)--(6,1)--(5,1)--cycle;
\draw[draw=none,fill=black](6,0)--(7,0)--(7,1)--(6,1)--cycle;
\draw[draw=none,fill=black](-5,1)--(-4,1)--(-4,2)--(-5,2)--cycle;
\draw[draw=none,fill=black](-4,1)--(-3,1)--(-3,2)--(-4,2)--cycle;
\draw[draw=none,fill=black](-3,1)--(-2,1)--(-2,2)--(-3,2)--cycle;
\draw[draw=none,fill=white](-2,1)--(-1,1)--(-1,2)--(-2,2)--cycle;
\draw[draw=none,fill=white](-1,1)--(0,1)--(0,2)--(-1,2)--cycle;
\draw[draw=none,fill=white](0,1)--(1,1)--(1,2)--(0,2)--cycle;
\draw[draw=none,fill=white](1,1)--(2,1)--(2,2)--(1,2)--cycle;
\draw[draw=none,fill=white](2,1)--(3,1)--(3,2)--(2,2)--cycle;
\draw[draw=none,fill=white](3,1)--(4,1)--(4,2)--(3,2)--cycle;
\draw[draw=none,fill=white](4,1)--(5,1)--(5,2)--(4,2)--cycle;
\draw[draw=none,fill=black](5,1)--(6,1)--(6,2)--(5,2)--cycle;
\draw[draw=none,fill=black](-4,2)--(-3,2)--(-3,3)--(-4,3)--cycle;
\draw[draw=none,fill=black](-3,2)--(-2,2)--(-2,3)--(-3,3)--cycle;
\draw[draw=none,fill=black](-2,2)--(-1,2)--(-1,3)--(-2,3)--cycle;
\draw[draw=none,fill=white](-1,2)--(0,2)--(0,3)--(-1,3)--cycle;
\draw[draw=none,fill=white](0,2)--(1,2)--(1,3)--(0,3)--cycle;
\draw[draw=none,fill=white](1,2)--(2,2)--(2,3)--(1,3)--cycle;
\draw[draw=none,fill=white](2,2)--(3,2)--(3,3)--(2,3)--cycle;
\draw[draw=none,fill=white](3,2)--(4,2)--(4,3)--(3,3)--cycle;
\draw[draw=none,fill=black](4,2)--(5,2)--(5,3)--(4,3)--cycle;
\draw[draw=none,fill=black](-3,3)--(-2,3)--(-2,4)--(-3,4)--cycle;
\draw[draw=none,fill=black](-2,3)--(-1,3)--(-1,4)--(-2,4)--cycle;
\draw[draw=none,fill=black](-1,3)--(0,3)--(0,4)--(-1,4)--cycle;
\draw[draw=none,fill=white](0,3)--(1,3)--(1,4)--(0,4)--cycle;
\draw[draw=none,fill=white](1,3)--(2,3)--(2,4)--(1,4)--cycle;
\draw[draw=none,fill=white](2,3)--(3,3)--(3,4)--(2,4)--cycle;
\draw[draw=none,fill=black](3,3)--(4,3)--(4,4)--(3,4)--cycle;
\draw[draw=none,fill=black](-2,4)--(-1,4)--(-1,5)--(-2,5)--cycle;
\draw[draw=none,fill=black](-1,4)--(0,4)--(0,5)--(-1,5)--cycle;
\draw[draw=none,fill=black](0,4)--(1,4)--(1,5)--(0,5)--cycle;
\draw[draw=none,fill=white](1,4)--(2,4)--(2,5)--(1,5)--cycle;
\draw[draw=none,fill=black](2,4)--(3,4)--(3,5)--(2,5)--cycle;
\draw[draw=none,fill=black](-1,5)--(0,5)--(0,6)--(-1,6)--cycle;
\draw[draw=none,fill=black](0,5)--(1,5)--(1,6)--(0,6)--cycle;
\draw[draw=none,fill=white](1,5)--(2,5)--(2,6)--(1,6)--cycle;
\draw[draw=none,fill=black](0,6)--(1,6)--(1,7)--(0,7)--cycle;
\draw[very thin]
(-6,0)--(-6,1)--(-5,1)--
(-5,1)--(-5,2)--(-4,2)--
(-4,2)--(-4,3)--(-3,3)--
(-3,3)--(-3,4)--(-2,4)--
(-2,4)--(-2,5)--(-1,5)--
(-1,5)--(-1,6)--(0,6)--
(0,6)--(0,7)--(1,7)--
(1,7)--(1,6)--(2,6)--
(2,6)--(2,5)--(3,5)--
(3,5)--(3,4)--(4,4)--
(4,4)--(4,3)--(5,3)--
(5,3)--(5,2)--(6,2)--
(6,2)--(6,1)--(7,1)--
(7,1)--(7,0)--cycle;
\end{tikzpicture}
\end{center}
\caption{A configuration of the fooling set $S_n$ for rule 218}
\label{fig:218_fool}
\end{figure}
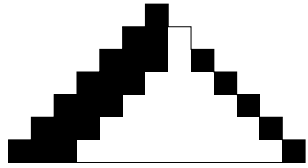
\end{proof}
\end{proposition}
\subsection{CA Rule 94}
\label{sec:94}

The local function 
$f_{94}: \{0,1\}^3 \to \{0,1\}$
of CA Rule 94 is defined in
Figure~\ref{fig:rule-94}.

\begin{figure}[htb]
\begin{center}
\subfigure[$f_{94}.$]{
\carule 0 1 1 1 1 0 1 0
\label{fig:rule-94}
}
\subfigure[Example of a space-time diagram for CA Rule 94.]
{
\docompile{\input{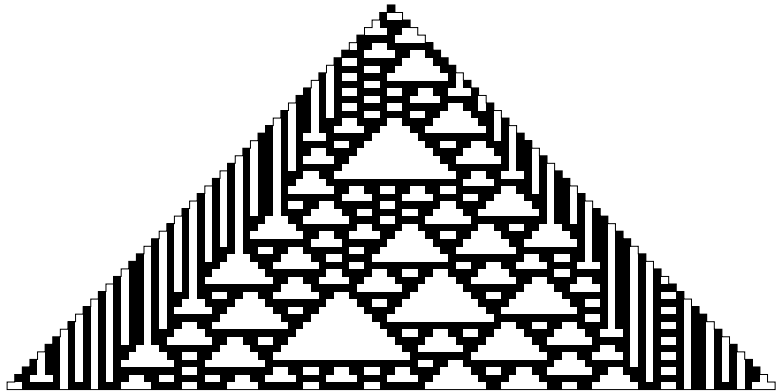}}
\label{spacetime-94}
}
\caption{CA Rule 94.}
\end{center}
\end{figure}

Here appears clearly how powerful the invasion approach is (as a tool
for proving non-universality).  Finding an upper bound (a protocol)
for $\cc{\predp{F_{94}}}$ seems to be
hard. Nevertheless, here we prove in a rather simple way
that its invasion complexity is
logarithmic.

\begin{definition}
  A configuration is \emph{additive} if its language is included in
  ${\bigl((00)^+(11)^+\bigr)^\ast}$ (blocks of 0s or 1s are always of
  even length).
\end{definition}

\begin{lemma}
  \label{94:additiv}
  $f_{94}$ is bi-permutative when restricted to additive
  configurations (it behaves like $f_{90}$) and additive
  configurations are stable under iterations.
\end{lemma}

\begin{proof}
  For stability of additive configurations, it is
  sufficient to check that $00(11)^n00$ becomes $11(00)^{n-1}11$ and
  $11(00)^n11$ becomes $11(00)^{n-1}11$ for $n\geq 1$.  
  
  $f_{94}$ differs from $f_{90}$ only for transition $010$, hence
  bi-permutativity. 
\end{proof}

\begin{lemma}
  \label{94:wall}
  If $c$ is a non-additive configuration which does not contain $010$,
  then $101$ appears after a finite time and it is a wall. More
  precisely, a wall appears after $t+1$ steps of CA Rule 94 at the middle of any
  occurrence of $10^{2t+1}1$ or ${01^{2t+3}0}$ (with $t\geq 0$).
\end{lemma}

\begin{proof}
  First $101$ is stable under iterations of $f_{94}$. Second, $10^n1$
  with $n\geq 2$ is sent to $10^{n-2}1$ and $01^n0$ is sent to
  $10^{n-2}1$ for $n\geq 2$.
\end{proof}

\begin{lemma}
  \label{94:teknik}
  The orbit of a configuration $c$ contains a wall if and only if
  $F_{94}(c)$ is not additive.
\end{lemma}

\begin{proof}
  From Lemma~\ref{94:wall}, it is enough to show that if $c$ is a
  configuration not containing $101$ then $F_{94}(c)$ does not contain
  $010$.  For that, it is sufficient to check that any word $u$ such
  that $f_{94}(u)=010$ must contain $101$.
\end{proof}

From the 2 lemmas above, we get the following proposition.

\begin{proposition}
  For any $u$ we have $\cc{\invap{F_{94}}{u}} \in O(\log(n))$.
\end{proposition}

\begin{proof}
  If $u$ is such that the orbit of $\perio{u}$ contains a wall, then
  invasion never occurs.

  If $u$ is such that the orbit of $\perio{u}$ does not contain any
  wall, then it means that $F_{94}(\perio{u})$ is additive (by
  Lemma~\ref{94:teknik}). In this situation, two
  cases are to be considered depending on the input ${x_1,\ldots,
    x_n}$. Knowing in which case we are can be done within constant
  number of bits:
  \begin{itemize}
  \item either ${F_{94}(\perio{u}(x_1,\ldots,x_n))}$ is also additive
    and then, by Lemma~\ref{94:additiv}, there is invasion if and only
    if ${F_{94}(\perio{u})=F_{94}(\perio{u}(x_1,\ldots,x_n))}$. This
    can be decided with a finite number of bits of communication.
  \item or ${F_{94}(\perio{u}(x_1,\ldots,x_n))}$ is not additive. Then
    it contains some $10^{2t+1}1$ or some ${01^{2t+3}0}$ (with $t\geq
    0$) because, as shown in the proof of lemma, if the image of a
    configuration contains $010$, then it must also contain $101$. Consider
    the leftmost and the rightmost occurrences of this kind of
    words. Since walls appear above the middle of these two occurrences
    after a time equal to half their lengths (Lemma~\ref{94:wall}),
    the fact there is invasion or not does not depend on what is
    beetween that two occurrences.  It takes $O(\log(n))$ bits of
    communication for Alice to know the positions of these two
    occurrences and the exact words present at their positions (of
    type $10^{2t+1}1$ or ${01^{2t+3}0}$). Moreover, as soon as Alice
    knows this she also knows that on the left of the leftmost
    occurrence and on the right of the rightmost occurrence, the
    configuration is additive. If there is one difference with
    $\perio{u}$ in those additive parts, then there is invasion. If
    not, then Alice has got enough information to decide invasion. Deciding
    in which of the two cases we are can be done within constant
    communication.
  \end{itemize}
\end{proof}

\begin{corollary}
CA Rule 94 is not intrinsically universal.
\end{corollary}

\subsection{CA Rule 33}
\label{sec:33}

We are going to show that this rule, although non-trivial for the
\predp{} problem, needs zero communication for the \cyclp{}{}
problem. To show this, we prove that the cycle length of Rule 33 is
always 2. The local function $f_{33}: \{0,1\}^3 \to \{0,1\}$ of CA
 Rule 33 is defined in Figure~\ref{fig:rule-33}.

\begin{figure}[ht!]
\centering
\subfigure[$f_{33}.$] 
{
{\carule 1 0 0 0 0 1 0 0}
\label{fig:rule-33}
}
\subfigure[Example of a space-time diagram for CA  Rule 33.] 
{
\docompile{\input{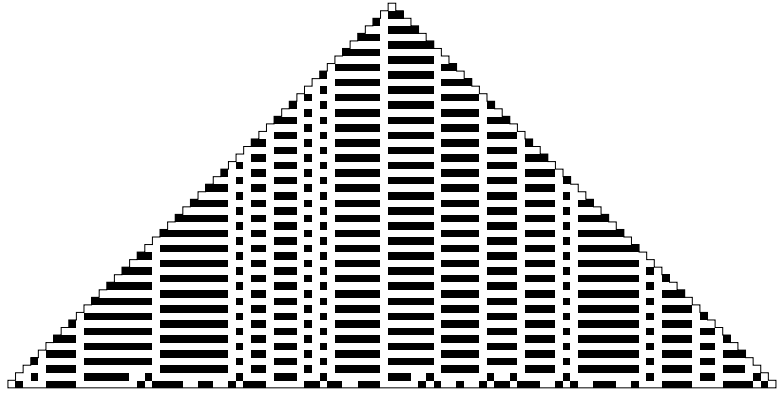}}
\label{spacetime-33}
}
\label{fig:33}
\caption{CA Rule 33.}
\end{figure}

\begin{lemma}
\label{lemma:33-1}
All configurations that do not contain neither $101$ (isolated $0$s) nor $1001$
(isolated $00$s) are stable under $(F_{33})^2$.
\begin{proof}

We call $A_0$ the set of configurations without isolated $0$s, and $A_{00}$ the
set configuration without isolated $00$s. First notice that the only antecedent
of $101$ is $10101$, which contains an isolated $0$, thus $A_{0}$ is stable
under $F_{33}$. With an exhaustive exploration of all configurations of the form
$u=abxyzcd$ where $xyz\in\{000\ldots 111\}$, and $u\in A_{0}\cap A_{00}$, we
observe that:
\[\forall u\in A_{0}\cap A_{00}, |u|=7, (F_{33})^2(u_1\ldots u_7)=u_3u_4u_5\]

\end{proof}
\end{lemma}

\begin{lemma}

All (cyclic) configurations of length $n$, different from $(01)^{\lfloor
n/2\rfloor}$, do not contain isolated $0$s after
$\left\lfloor\frac{n}{2}\right\rfloor$ steps of CA Rule 33.

\begin{proof}

We already noticed in Lemma~\ref{lemma:33-1} that the only possible antecedent
of $101$ is $10101$. Thus, there can be an isolated $0$ after
$\left\lfloor\frac{n}{2}\right\rfloor$ steps only if there are at least
$\left\lfloor\frac{n}{2}\right\rfloor$ isolated $0$s in the initial
configuration, i.e. if the initial configuration is $(01)^{\lfloor n/2\rfloor}$.

\end{proof}
\end{lemma}
\begin{corollary}
After $\left\lfloor\frac{n}{2}\right\rfloor+1$ steps of CA Rule 33, 
there are no isolated couples of $0$s.
\begin{proof}
The only antecedents of $1001$ contain an isolated $0$.
\end{proof}
\end{corollary}

\begin{corollary}
After $\left\lfloor\frac{n}{2}\right\rfloor+1$ steps, CA Rule 33 becomes periodic,
with period $2$. 
\end{corollary}

\begin{proposition}
$$\cc{\predp{F_{33}}}\in\Omega(\log n)$$


\begin{proof}
As usual, we just find a fooling set (see Definition~\ref{def:fool}). 
Consider the following set $S_n$:
\[S_n=\{(1^{n-2k}(01)^{k}0,(10)^{k}1^{n-2k}),0\leq k\leq \lfloor n/2\rfloor \}\]

It can be easily verified that:
\[\left\{\begin{array}{lcll}
F_{33}^n(1^{2n-k}(01)^{k}0(10)^{k}1^{2n-k})&=&n\mod 2&\\
F_{33}^n(1^{2n-i}(01)^{i}0(10)^{j}1^{2n-j})&=&1+(n\mod 2)&\text{whenever }i\neq j
\end{array}\right.\]

Since $|S_n|=\left\lfloor\frac{n}{2}\right\rfloor$, we conclude that a
deterministic protocol for predicting rule 33 needs $\Omega(\log n)$ 
bits of communication.

\end{proof}
\end{proposition}

\section{Conclusion}

We have suggested a method to prove negative results concerning
intrinsic universality in CA. We have shown that this
approach can be used both to show that global dynamical
properties can imply non-universality, and to rule out some concrete cellular
automata from being universal. We believe that this work should
go on in the following directions:

\begin{itemize}
\item It seems that more can be said about the communication
  complexity problems for the class of surjective CA and some of its
  sub-classes ($k$-to-1, $d$-separated, left/right-closing, etc.~\cite{hedlund69});
\item The case of elementary rules $218$ and $94$ shows that low-cost
  communication protocols can be found in CA that are not linear,
  but containing a linear component `in competition' with another
  component. Finding a general formalisation for such kind of
  behaviours could be useful to treat many other concrete examples.
\item Concerning concrete CA, ruling out as many elementary rules as
  possible from being intrinsically universal seems to be an
  interesting (but ambitious) goal. We could also consider other
  natural classes of small CA (one-way automata, totalistic rules,
  etc.).
\item The splitting of inputs that induce maximal
  communication complexity is a key parameter, especially
  for the prediction problem. There is no reason for such maximal
  splittings to be unique, and if it is unique, there is no reason
  to be located in the  middle of the input. We suspect that there are
  some links between directional entropy and the evolution of such
  maximal splitting (when increasing the input size).
\item Although completely formalized in dimension $1$, there is no
  doubt that this approach can be adapted to higher dimensions; it
  could be the occasion to adopt other communication complexity models
  (like the multiparty model) and discuss other ways of splitting
  the input.
\end{itemize}

\bibliographystyle{plain}
\bibliography{main}

\end{document}